\documentclass[11pt,a4paper]{article}
\usepackage[utf8]{inputenc}
\usepackage[margin=1in]{geometry}
\usepackage{amsmath,amssymb,amsthm}
\usepackage{xspace}
\usepackage{complexity}
\usepackage{xcolor}
\usepackage{mathpazo}
\usepackage{hyperref}
\hypersetup{colorlinks=true,linkcolor=blue,citecolor=[rgb]{0.1,0.1,0.9}}
\usepackage{tcolorbox}
\usepackage{thm-restate}
\usepackage{booktabs}
\usepackage{multirow}
\usepackage{tabularx}
\usepackage{array}

\usepackage{setspace}
\newtheorem{theorem}{Theorem}[section]
\newtheorem{cor}[theorem]{Corollary}
\newtheorem{lemma}[theorem]{Lemma}

\newtheorem{claim}[theorem]{Claim}
\newtheorem{remark}[theorem]{Remark}

\newtheorem{obs}[theorem]{Observation}
\newtheorem{defi}[theorem]{Definition}
\newtheorem{hypo}[theorem]{Hypothesis}
\newtheorem{fact}[theorem]{Fact}

\renewcommand{\bar}{\overline}

\renewcommand{\tilde}{\widetilde}
\DeclareMathOperator*\bigcircop{\bigcirc}
\newcommand{\cm}{\textsc{Continuous Ulam Median}\xspace}

\newcommand{\dc}{\textsc{Discrete Ulam Center}\xspace}
\newcommand{\dm}{\textsc{Discrete Ulam Median}\xspace}
\newcommand{\LCS}{\mathsf{LCS}}

\newcommand{\SETH}{\textsf{SETH}\xspace}
\newcommand{\fourOVH}{{Unbalanced} \(4\)-\textsf{OVH}\xspace}

\newcommand{\eaeSETH}{\(\exists\forall\exists\)\textsf{SETH}\xspace}
\newcommand{\ueaeeOVH}{{Unbalanced} \(\exists\forall\exists\exists\)\textsf{OVH}\xspace}

\title{\vspace{-1cm}\textbf{Hardness of Median and Center in the Ulam Metric}}
\author{ Nick Fischer\footnote{INSAIT, Sofia University ``St.\ Kliment Ohridski''. Email: \texttt{nick.fischer@insait.ai}} \and Elazar Goldenberg\footnote{The  Academic College of Tel Aviv-Yaffo. Email: \texttt{elazargo@mta.ac.il}} \and Mursalin Habib\footnote{Rutgers University. Email: \texttt{mursalin.habib@rutgers.edu}} \and Karthik C.\ S.\footnote{Rutgers University. Email: \texttt{karthik.cs@rutgers.edu}}} 
\date{}

\begin{document}

\maketitle
\begin{abstract}
\noindent
The classical \emph{rank aggregation} problem seeks to combine a set $X$ of $n$ permutations into a single representative ``consensus'' permutation. In this paper, we investigate two fundamental rank aggregation tasks under the well-studied \emph{Ulam metric}: computing a \emph{median} permutation (which minimizes the sum of Ulam distances to $X$) and computing a \emph{center} permutation (which minimizes the maximum Ulam distance to $X$) in two settings. 
\begin{description}
    \item[Continuous Setting:] In the continuous setting, the median/center is allowed to be any permutation. It is known that computing a center in the Ulam metric is NP-hard and we add to this by showing that computing a median is NP-hard as well via a simple reduction from the Max-Cut problem.  While this result may not be unexpected, it had remained elusive until now and confirms a speculation by Chakraborty, Das, and Krauthgamer~\makebox{[SODA '21]}.
    
    \item[Discrete Setting:] In the discrete setting, the median/center must be a permutation from the input set. We fully resolve the fine-grained complexity of the discrete median and discrete center problems under the Ulam metric, proving that the naive~\smash{$\tilde{O}(n^2 L)$}-time algorithm (where $L$ is the length of the permutation) is conditionally optimal. This resolves an open problem raised by Abboud, Bateni, Cohen-Addad, Karthik C.\ S., and Seddighin [APPROX '23]. Our reductions are inspired by the known fine-grained lower bounds for similarity measures, but we face and overcome several new highly technical challenges.
\end{description}
\end{abstract}

\section{Introduction}
Suppose that $n$ judges each rank the performances of $L$ competitors. Given these rankings, how can the judges agree on a single consensus ranking? This fundamental question lies at the heart of a class of tasks known as \emph{rank aggregation}, which has applications across various fields, including social choice theory~\cite{BrandtCELP16}, bioinformatics~\cite{Li19}, information retrieval~\cite{Harman92a}, machine learning~\cite{LiuLQML07}, and recommendation systems~\cite{OliveiraDLMP20}, among others. Formally, the judges' rankings can be represented as a set of $n$ permutations~\makebox{$X \subseteq \mathcal S_L$}. Then, for an appropriate metric~$d(\cdot,\cdot)$ on the space of permutations $\mathcal S_L$, the two most prominent rank aggregation tasks are to compute a \emph{median} permutation~$\pi_M$ which minimizes $\sum_{\pi \in X} d(\pi_M, \pi)$~\cite{Kemeny59,Young88,YoungL78,DworkKNS01}, or a \emph{center} permutation $\pi_C$ which minimizes $\max_{\pi \in X} d(\pi_C, \pi)$~\cite{BachmaierBGH15,BiedlBD09,Popov07}. 

Among the metrics studied in this context, two stand out. The first one is the classic \emph{Kendall's tau distance} which measures the number of disagreeing pairs between two permutations, i.e., the number of pairs $(i, j)$ for which one ranking orders \(i\) before \(j\) while the other orders $j$ before $i$. Kendall's tau distance is well-motivated as it satisfies several desirable properties beyond the scope of this paper (e.g., neutrality, consistency, and the extended Condorcet property~\cite{Kemeny59,WangSCK12}), and is also well-understood from a computational point-of-view~\cite{DworkKNS01,FaginKS03}. For example, it is known that computing the median or center of just \emph{four} permutations is already NP-hard~\cite{DworkKNS01,BiedlBD09}. Several approximation algorithms have also been proposed for this metric~\cite{AilonCN08,ZuylenW07}, culminating in a PTAS~\cite{Kenyon-MathieuS07,Pevzner00} for approximating the median under Kendall's tau metric.

The other key metric is the \emph{Ulam distance} which measures the minimum number of relocation operations required to turn one permutation $\pi$ into another permutation $\pi'$ -- i.e., the minimum number of competitors whose ranks have to be adjusted in $\pi$ so that it agrees with $\pi'$. This metric offers a simpler and more practical alternative to Kendall's tau metric for rank aggregation tasks~\cite{CormodeMS01,ChakrabortyDK21,ChakrabortyGJ21,ChakrabortyD0S22}. Perhaps more importantly, the Ulam metric is intimately linked to the more general \emph{edit metric} on arbitrary strings, which enjoys countless applications in computational biology~\cite{Gusfield1997,Pevzner00}, specifically in the context of DNA storage systems~\cite{GoldmanBCDLSB13,RashtchianMRAJY17}, and beyond~\cite{Kohonen85,Martinez-HinarejosJC00}. Despite the significance of Ulam rank aggregation problems and the extensive research dedicated to them~\cite{BachmaierBGH15,ChakrabortyDK21,ChakrabortyGJ21,ChakrabortyD0S22}, some basic questions remain unanswered. This is the starting point of our paper.

\subsection{Question 1: Polynomial-Time Algorithms for Ulam Median?}
The first basic question is whether polynomial-time algorithms exist for exactly computing the center and median permutations under the Ulam metric.
 For almost all string metrics, including the aforementioned Kendall's tau metric but also metrics beyond permutations such as the Hamming metric~\cite{FrancesL97,LanctotLMWZ03,AbboudFGSS23} or the edit metric~\cite{HigueraC00,NicolasR03}, median and center problems are well-known and easily-proven to be NP-hard.\footnote{A notable exception is the Hamming median problem that can trivially be solved in polynomial time by a coordinate-wise plurality vote.} Quite surprisingly, while it is known that computing an Ulam center is NP-hard~\cite{BachmaierBGH15}, the complexity of computing an Ulam median has remained an open question. This is not due to a lack of interest -- despite the absence of an NP-hardness proof, Chakraborty, Das, and Krauthgamer have already initiated the study of approximation algorithms for the Ulam median problem~\cite{ChakrabortyDK21,Chakraborty0K23}, achieving a 1.999-factor approximation in polynomial time. Our first contribution is that we finally provide this missing hardness proof:

\begin{theorem}
\label{thm:result-1}
    The median problem is NP-hard in the Ulam metric.
\end{theorem}


\subsection{Question 2: Fine-Grained Complexity of Discrete Ulam Center and Median?}
How can we circumvent this new lower bound? There are two typical approaches. The first is to resort to approximation algorithms as was done in~\cite{ChakrabortyDK21,Chakraborty0K23}. But there is a commonly-studied second option for aggregation and clustering type problems: to restrict the solution space to \textit{only the input set of permutations} \(X\) and compute the best median or center from it.\footnote{Or alternatively, the best median/center among the permutations in another given set $Y \subseteq \mathcal S_L$; this is typically referred to as the \emph{bichromatic} discrete center/median problem, or also as center/median problem with \emph{facilities} in the theory of clustering. All of our results also apply to these bichromatic variants.\label{foot:bichromatic}} The best median or center from the input set \(X\) is typically referred to as the \emph{discrete median} or the \emph{discrete center}, respectively, of \(X\). (In the same spirit, we will occasionally refer to the unrestricted median and center, discussed in the previous subsection, as the \emph{continuous median} and \emph{continuous center}, respectively.) Besides being a natural polynomial-time rank aggregation task, computing discrete medians or centers has two other motivations. First, it is easy to see that computing a discrete median/center yields a 2-approximation of the (continuous) median/center problems. In fact, a key observation in~\cite{ChakrabortyDK21} is that a discrete median often provides a 
\((2-\varepsilon)\)-approximation for the (continuous) median, particularly in practical DNA storage system instances where distances tend to be large. Second, the discrete median and center problems have gained significant attraction for the easier Hamming metric~\cite{AbboudBCSS23,AbboudFGSS23} and harder edit metric~\cite{AbboudBCSS23}, often leading to matching upper and lower bounds. Studying the Ulam metric therefore serves as an interesting intermediate problem capturing some -- but not all -- of the hardness of the edit metric.

Driven by these motivations, we study the \emph{fine-grained complexity} of the discrete median and discrete center problems with respect to the Ulam distance. That is, we aim to pinpoint their \emph{precise} polynomial run times.
\paragraph{Discrete Ulam Center.}
The trivial algorithm for computing a discrete center is to explicitly compute the Ulam distance $d_U(\pi, \pi')$ for all pairs of permutations $\pi, \pi' \in X$. We can then easily select the permutation~\makebox{$\pi_C \in X$} minimizing $\max_{\pi \in X} d_U(\pi_C, \pi)$. As the Ulam distance between two length-$L$ permutations can be computed in near-linear time\footnote{We write $\tilde O(T) = T (\log T)^{O(1)}$ to suppress polylogarithmic factors.} $\tilde O(L)$ using the well-known patience sorting algorithm~\cite{AldousD99}, the total time is $\tilde O(n^2 L)$.

We prove that this simple algorithm is \emph{optimal}, up to subpolynomial factors and conditioned on a plausible assumption from fine-grained complexity:

\begin{theorem}
    \label{thm:result-3}
    Let $\varepsilon > 0$ and $\alpha > 0$. There is no algorithm running in time $O((n^2L)^{1-\varepsilon})$ that solves the discrete center problem in the Ulam metric for $n$ permutations of length $L = \Theta(n^{\alpha})$, unless the Quantified Strong Exponential Time Hypothesis fails.
\end{theorem}

The Quantified Strong Exponential Time Hypothesis (QSETH) is a plausible generalization of the by-now well-established Strong Exponential Time Hypothesis (SETH), postulating that the CNF-SAT problem cannot be solved faster than brute-force search even when only some variables are existentially quantified and others are universally quantified (see Section~\ref{sec:prelims-hardness} for the formal treatment)~\cite{BringmannC20}. This hypothesis has already proven useful for conditional lower bounds for a wide array of problems~\cite{BringmannC20,AbboudBHS22,AbboudBCSS23}. Besides, we remark that it is impossible to obtain SETH-based lower bounds for the discrete center problem unless the \emph{Nondeterministic} Strong Exponential Time Hypothesis~\cite{CarmosinoGIMPS16} is false (see Section~\ref{sec:quant-justification} for details). 

We emphasize that these lower bounds applies to the full range of $n$ versus $L$, as long as~$L$ is at least polynomial in $n$. In the case when $L$ is very small, $\omega(\log n) < L < n^{o(1)}$, previous work by Abboud, Bateni, Cohen-Addad, Karthik C.\ S.\ and Seddighin already established a matching conditional lower bound of~\makebox{$n^{2-o(1)}$}~\cite{AbboudBCSS23}.

It is interesting to compare Theorem~\ref{thm:result-3} with the state of the art for discrete center problems in the Hamming metric (say, over a constant-size alphabet) and the edit metric. For concreteness, consider the case $L = \Theta(n)$ (i.e., the input consists of $n$ strings/permutations of length roughly~$n$). Then, on the one hand, the discrete center problem in the Hamming metric can be solved in time~$O(n^\omega)$~\cite{AbboudBCSS23,AbboudFGSS23}, where $\omega < 2.3714$ is the exponent of matrix multiplication~\cite{AlmanDWXXZ25}. On the other hand, the discrete center problem in the edit metric cannot be solved in time $n^{4-\Omega(1)}$ unless  QSETH fails~\cite{AbboudBCSS23}. Therefore, remarkably, Theorem~\ref{thm:result-3} indeed places the discrete center problem in the Ulam metric as a problem of \emph{intermediate} complexity~$n^{3\pm o(1)}$. This answers an explicit open question posed by Abboud, Bateni, Cohen-Addad, Karthik~C.\ S., and Seddighin~\cite{AbboudBCSS23}.

\paragraph{Discrete Ulam Median.}
The trivial algorithm for the discrete median problem is exactly the same as for the center problem: First compute all pairwise distances $d_U(\pi, \pi')$, then select the permutation $\pi_M$ minimizing $\sum_{\pi \in X} d_U(\pi_M, \pi)$. It also runs in time $\tilde O(n^2 L)$. So perhaps one could also hope that a matching lower bound follows from our Theorem~\ref{thm:result-3}. Unfortunately, this turns out to be true only for a restricted subproblem.\footnote{Namely, the \emph{bichromatic} discrete median problem mentioned in Footnote~\ref{foot:bichromatic}.} Nevertheless, with considerable technical overhead we manage to prove essentially the same \emph{matching} lower bound:

\begin{theorem}
    \label{thm:result-4}
    Let $\varepsilon > 0$ and $\alpha \geq 1$. There is no algorithm running in time $O((n^2L)^{1-\varepsilon})$ that solves the discrete median in the Ulam metric for $n$ permutations of length $L = \Theta(n^{\alpha})$, unless the SETH fails.
\end{theorem}

In comparison to Theorem~\ref{thm:result-3}, Theorem~\ref{thm:result-4} has the advantage that it conditions on the \emph{weaker} assumption SETH (thus constituting a \emph{stronger} lower bound). However, the applicable range of parameters is more restricted ($\alpha \geq 1$ forces the permutations to have length at least $\Omega(n)$).

\section{Proof Overview}
In this section, we provide the proof overviews of Theorems~\ref{thm:result-1}, \ref{thm:result-3}, and \ref{thm:result-4}. We avoid some technicalities to highlight the core ideas -- for instance, some subpolynomial factors have been dropped. Complete proofs can be found in Sections~\ref{sec:continuous-median},~\ref{sec:discrete-center}, and~\ref{sec:median}, respectively.

\subsection{NP-Hardness for Continuous Median in the Ulam Metric}
In this section, we provide a high-level sketch of our NP-hardness proof for the continuous median problem in the Ulam metric. Our starting point is the Max-Cut problem\footnote{Recall that the Max-Cut problem is, given an undirected graph $G = (V, E)$ to compute a vertex partition $V = A \sqcup B$ maximizing the number of edges from $A$ to $B$.}. Given a Max-Cut instance \(G = (V, E)\) where \(V = [n]\), our goal is to construct a set of permutations of length \(O(n)\) such that the median of these permutations encodes the cut in \(G\) of maximum size.

To achieve this, we first set up a natural correspondence between cuts in \(G\) and permutations of length \(O(n)\). However, since there are many more permutations than cuts, not all permutations will represent valid cuts. To ensure that only relevant permutations are considered, we construct two special permutations:
\begin{align*}
    \pi^L &= 1 \circ 2 \circ \cdots \circ (n-1) \circ n \circ X_1 \circ X_2, \\
    \pi^R &= X_1 \circ n \circ (n-1) \circ \cdots \circ 2 \circ 1 \circ X_2,
\end{align*}
for some fixed long permutations $X_1, X_2$. The simple key insight is that any median of \(\pi^L\) and~\(\pi^R\) has the following form: it starts with some subset $A \subseteq [n]$ of the symbols in \emph{increasing} order, followed by $X_1$, followed by the symbols in $[n] \setminus A$ in \emph{decreasing} order, finally followed by $X_2$. Thus, medians of \(\pi^L, \pi^R\) will naturally encode cuts of the form \((A, [n]\setminus A)\) in~\(G\).

To further enforce that the median represents a \textit{maximum} cut, we include additional permutations in our instance. Specifically, for each edge \(e\in E\), we include two permutations, \(\pi_e^1\) and \(\pi_e^2\), which reward picking solutions that correspond to partitions that cut \(e\). This ensures that the final median permutation encodes a maximum cut of \(G\). The precise construction and formal analysis of these edge-cutting permutations \(\pi_e^1, \pi_e^2\) is detailed in Section~\ref{sec:continuous-median}.

\subsection{Fine-Grained Lower Bound for Discrete Center in the Ulam Metric}
Our proof of the lower bound for the discrete center problem in the Ulam metric relies on two key components. The first is a reduction from the Orthogonal Vectors (OV) problem to the problem of computing the Ulam distance between two permutations. Specifically, we seek a pair of functions that, given two sets of binary vectors as inputs, independently output two permutations whose Ulam distance is small if and only if there exists an orthogonal pair of vectors in the input sets. This falls into a well-established framework in the fine-grained complexity literature~\cite{BackursI18,BringmannK15,AbboudBW15}: Given two sets of roughly $L$ binary vectors, one constructs ``coordinate gadgets'', ``vector gadgets'', and ``OR-gadgets'' to produce a pair of length-\(L\) strings whose edit distance encodes the existence of an orthogonal pair.

Clearly such a reduction cannot exist for the Ulam distance under SETH. The Ulam distance between two permutations can be computed in near-linear time by a simple reduction to the longest increasing subsequence problem~\cite{schensted1961longest}, and thus, if there were a way to transform sets of \(O(L)\) many vectors into permutations of length \(L\) such that the existence of an orthogonal pair in the sets could be determined via an Ulam distance computation of these permutations, then that would imply a near-linear time algorithm for OV, falsifying SETH! In light of this observation, we start with \(O(\sqrt{L})\) vectors in the OV instance. The constructions of the coordinate and vector gadgets are similar to the edit distance reduction. However, during the construction of the OR-gadgets, there is a quadratic blowup resulting in length \(L\) permutations with the desired properties. A detailed construction of these gadgets can be found in Section~\ref{sec:discrete-center}.

The second component in our proof is a reduction from a problem called \(\exists\forall\exists\exists\)-Orthogonal Vectors. In this problem, we are given four sets \(A, B, C, E\) of binary vectors and we have to decide if there exists \(a\in A\), such that for all \(b\in B\), there exist \(c\in C, e\in E\) such that \(a, b, c, e\) are orthogonal\footnote{We say that vectors \(a, b, c, e\in \{0, 1\}^d\) are \emph{orthogonal} if \(\sum_{i\in [d]}a[i]b[i]c[i]e[i]=0.\)}. If \(|A|=|B|=n\) and \(|C|=|E|=\sqrt{L}\), then this problem has a \(O((n^2L)^{1-\Omega(1)})\) lower bound under the Quantified Strong Exponential Time Hypothesis. Given such sets \(A, B, C, E\), we proceed as follows. First, for each \(a\in A\), we construct the set \(V_a\) of \(\sqrt{L}\) vectors by taking the pointwise product of \(a\) with all \(\sqrt{L}\) vectors in \(C\). There will be \(n\) such sets \(V_a\), one for each choice of \(a\in A\). Similarly, for each \(b\in B\), we construct the set \(W_b\) of \(\sqrt{L}\) vectors by taking the pointwise product of \(b\) with all \(\sqrt{L}\) vectors in \(C\). We then run the OV to Ulam distance reduction from before on these sets to obtain two sets of \(n\)-many permutations of length \(L\). Finally, we show that there exists a permutation in the first set with small Ulam distance to every permutation in the second set if and only if the starting \(\exists\forall\exists\exists\)-OV instance is a YES-instance.

To go from these two sets to the final discrete center instance, we append additional symbols to each permutation and introduce a new permutation that is far from every permutation in the second set. This ensures that the center indeed comes from the first set completing the reduction. We defer the details to Section~\ref{sec:discrete-center}.

\subsection{Fine-Grained Lower Bound for Discrete Median in the Ulam Metric}
Our lower bound proof for the discrete median problem follows a similar initial approach as our proof for the center lower bound, with only one difference: instead of starting with an \(\exists\forall\exists\exists\)-OV instance, we begin with a \(\exists\exists\exists\exists\)-OV (also known as 4-OV) instance. Given this 4-OV instance, we retrace the same steps to construct two sets, \(X\) and \(Y\), each containing \(n\) permutations of length \(L\). As in the center proof, we show that there exists a permutation in \(X\) whose total Ulam distance to all elements in \(Y\) is small if and only if the original 4-OV instance is a YES-instance.  

However, going from these two sets to the standard single-set version of the discrete median problem  is technically very challenging. In fact, such challenges were addressed in the past in the context of the closest pair problem \cite{DKL19,KM20}, and more generally identified as the task of \emph{reversing color coding} \cite{BKN21}, typically requiring extremal combinatorial objects which are then composed with the input in a black-box manner. 

In this paper, we transform the bichromatic instance to a monochromatic one, in multiple steps but in a white-box manner using the structure of the input instance. The first key observation is that all pairwise Ulam distances within \(X\) can be computed much faster than the naive \(O(n^2L)\) time bound, specifically, in \(O(n^2\sqrt{L})\) time. This speedup is possible because the permutations in \(X\) are not arbitrary but outputs formed during our OV to Ulam distance reduction. Thus, in \(O(n^2\sqrt{L})\) time, we can compute the total Ulam distance of each \(x \in X\) to all other elements in \(X\).  

Once these \(n\) distance sums are computed, we initiate a \textit{balancing} procedure. This procedure iteratively appends additional symbols to each permutation in \(X \cup Y\) such that:
\begin{itemize}
    \item For every permutation in \(X\), the sum of its Ulam distances to all other elements in \(X\) becomes equal.
    \item The relative Ulam distances between permutations across the sets remain unchanged.
    \item For every permutation in \(Y\), the sum of its Ulam distances to all other elements in \(Y\) becomes very large.
\end{itemize}
We show that this balancing procedure can be performed efficiently without significantly increasing the permutation lengths. Finally, we include all modified permutations into a single set, and output that as our final discrete median instance. The details turn out to be quite involved, and we direct the reader to Section~\ref{sec:median} for further details.

    \section{Preliminaries}

\paragraph{Sets, Strings and Permutations.} For a positive integer \(n\), let \([n]\) denote the set \(\{1, 2, \ldots, n\}\). We denote by \(\mathcal{S}_n\) the set of all permutations over \([n]\). Throughout the paper, we treat any permutation \(s\in \mathcal{S}_n\) as a string over the alphabet \([n]\) with no repeating symbols, and we write \(s:= s[1]s[2]\ldots s[n]\). Given \(k\) strings \(s_1, s_2, \ldots , s_k\) over some alphabet, we denote by \(\bigcircop_{i\in [k]}s_i\) the concatenated string \(s_1s_2\ldots s_k\). We will often require strings that can be obtained by adding some fixed ``offset'' to each symbol of some other canonical string\footnote{Here we are treating the symbols of a string as integers themselves.}. For every nonnegative integer \(k\) and string of length \(n\), we let \(\Delta_k(s) := (s[1]+k)(s[2]+k)\ldots(s[n]+k)\). In other words, \(\Delta_k(s)\) is the string obtained by adding \(k\) to each symbol of \(s\). We will also often require strings that are obtained by restricting some string to some sub-alphabet. Given any string \(s\) over the alphabet \(\Sigma\) and any sub-alphabet \(\Sigma'\subseteq \Sigma\), we denote by \(s|_{\Sigma'}\) to be string obtained from~\(s\) by deleting all symbols that are not in \(\Sigma'\). Given \(a, b\in \{0, 1\}^d\), we write \(\langle a, b\rangle\) for the inner product of \(a\) and \(b\), i.e., \(\langle a, b\rangle= \sum_{i\in [d]} a[i]b[i]\). We further write \(a\odot b\) for the pointwise product \(a\) and \(b\), i.e., \(a\odot b\in \{0, 1\}^d\) is the vector satisfying \((a\odot b)[i]=a[i]b[i]\) for all \(i\in [d]\).

\paragraph{Ulam Distances and Common Subsequences.} The Hamming distance between two equal-length strings \(x\) and \(y\), denoted by \(d_H(x, y)\), is the number of locations where \(x\) and \(y\) have different symbols. Let \(\pi:= \pi[1]\pi[2]\ldots\pi[n]\in\mathcal{S}_n\) be a permutation and \(i, j\in [n]\) be distinct positions. A \textit{symbol relocation} operation from position \(j\) to position \(i\) in \(\pi\) constitutes of deleting the \(j^{\text{th}}\) symbol of \(\pi\) and reinserting it at position \(i\). More formally, if \(\Tilde{\pi}\in \mathcal{S}_n\) is the permutation obtained after applying a symbol relocation from position \(j\) to position \(i\) in \(\pi\), then:
\[\Tilde{\pi} := \begin{cases}
            \pi[1]\pi[2]\ldots \pi[j-1]\pi[j+1]\cdots \pi[i-1]\pi[j]\pi[i]\pi[i+1]\cdots \pi[n], &\text{ if } j < i,\\
            \pi[1]\pi[2]\cdots \pi[i-1]\pi[j] \pi[i]\pi[i+1]\cdots \pi[j-1]\pi[j+1]\cdots \pi[n], &\text{ if } j > i.
        \end{cases}\]
Given \(\pi, \pi'\in \mathcal{S}_n\), the Ulam distance between \(\pi\) and \(\pi'\), denoted by \(d_U(\pi, \pi')\), is the minimum number of symbol relocation operations required to transform \(\pi\) into \(\pi'\). We will further denote by \(\LCS(x, y)\) the length of a longest common subsequence of two strings \(x\) and \(y\). We will frequently use the following fact throughout the paper, which relates the Ulam distance between two permutations to the length of their longest common subsequence.
\begin{fact}[\cite{AldousD99}]
\label{fact:ulam-lcs}
    For every \(\pi, \pi'\in \mathcal{S}_n\), we have \(d_U(\pi, \pi') = n - \LCS(\pi, \pi')\).
\end{fact}

\subsection{Hardness Assumptions} \label{sec:prelims-hardness}

Our results are conditional on several hardness assumptions and hypotheses, which we list in this section. The first one is the Strong Exponential Time Hypothesis, which is a standard assumption in the theory of fine-grained complexity.

\begin{hypo}[Strong Exponential Time Hypothesis (\SETH)]
For all \(\varepsilon >0\), there exists \(q\geq 3\) such that no algorithm running in time \(O(2^{(1-\varepsilon)n})\) can solve the \(q\)-SAT problem on \(n\) variables.
\end{hypo}
More specifically, for one of our lower bounds, we will need the following corollary of \SETH, which we dub \fourOVH.

\begin{hypo}[\fourOVH]
    For all \(\varepsilon>0\), no algorithm can, given sets \(A, B, C, E\subseteq \{0, 1\}^d\) with \(|A|=n\), \(|B|, |C|, |E| = n^{\Theta(1)}\) and \(\omega(\log n) < d < n^{o(1)}\), determine if there exists $a\in A$, $b\in B$, $c\in C$, $e\in E$ such that \(\sum_{i\in [d]}a[i]b[i]c[i]e[i]=0\) in time \(O((|A||B||C||E|)^{1-\varepsilon})\). 

\end{hypo}
It is well-known that \SETH in conjunction with the sparsification lemma~\cite{impagliazzo2001problems} implies \fourOVH~\cite{Williams05}. We will also make use of the following strengthening of \SETH.

\begin{hypo}[\eaeSETH]
    For all \(\varepsilon >0\) and \(0< \alpha <\beta <1\), there exists \(q\geq 3\), such that given a \(q\)-CNF formula \(\phi\) over the variables \(x_1, x_2, \ldots , x_n\), no algorithm running in time \(O(2^{(1-\varepsilon)n})\) can determine if the following is true:
\[\exists x_1, \ldots , x_{\lceil \alpha n\rceil}\forall x_{\lceil \alpha n\rceil+1}, \ldots , x_{\lceil\beta n\rceil} \exists x_{\lceil \beta n\rceil+1}, \ldots , x_{n} \ \phi(x_1, x_2, \ldots, x_n).\]
\end{hypo}
We note that \eaeSETH is a special case of the Quantified SETH proposed by Bringmann and Chaudhury~\cite{BringmannC20} -- a hypothesis about the complexity of deciding quantified \(q\)-CNF formulas with a constant number of quantifier blocks where each block contains some constant fraction of the variables. We do not formally define Quantified SETH in all of its generality since we only require three quantfier alternations. In fact, the specific hardness assumption we need is the following which is implied by \eaeSETH.

\begin{hypo}[\ueaeeOVH]
    For all \(\varepsilon>0\), no algorithm can, given sets \(A, B, C,\) \(E\subseteq \{0, 1\}^d\) with \(|A|=n\), \(|B|, |C|, |E| = n^{\Theta(1)}\) and \(\omega(\log n) < d < n^{o(1)}\), determine if there exists  \(a\in A\) such that for all \(b \in B\), there exist \(c\in C, e\in E\) such that \(\sum_{i\in [d]}a[i]b[i]c[i]e[i]=0\) in time \(O((|A||B||C||E|)^{1-\varepsilon})\). 
\end{hypo}

\section{NP-Hardness of Continuous Median in the Ulam Metric}
\label{sec:continuous-median}
In this section, we prove Theorem~\ref{thm:result-1}. Before we do so, we first formally define the continuous median problem in the Ulam metric.

\begin{table}[!h]
    \centering
    \renewcommand{\arraystretch}{1.4} 

    \begin{tabular}{|p{0.927\textwidth}|}
        \hline
        \cm \\
        \hline
    \end{tabular}

    \begin{tabular}{|p{0.15\textwidth}|p{0.75\textwidth}|}
        \hline
        \textbf{Input:} & A set \( S \subseteq \mathcal S_n \) of permutations and an integer \( d \). \\
        \hline
        \textbf{Question:} & Is there a permutation \( \pi^* \in S_n \) such that \( \sum_{\pi \in S} d_U(\pi, \pi^*) \leq d \)? \\
        \hline
    \end{tabular}

\end{table}

\noindent The main result of this section is the following.

\begin{theorem}
    \cm is NP-hard.
\end{theorem}
\begin{proof}
    We will reduce from the \textsc{Max Cut} problem, which is NP-hard~\cite{Karp72}. Let \(G=(V, E)\) be a \textsc{Max Cut} instance with \(V = [n]\). From \(G\), we will construct a \cm instance \(S \subseteq \mathcal{S}_{N}\) consisting of permutations of length \(N := 3n+2\). On a high level, our construction will work as follows. We will first construct many copies of two special permutations that will force every median of \(S\) to take on a very specific structure. All permutations of this structure will naturally encode cuts of the vertex set \(V\). Then for each edge \(e\) in \(G\), we will construct permutations that reward choosing a median that ``cuts'' \(e\). Thus, we will end up with a set of permutations whose overall median will encode the maximum cut of \(G\). Details follow.

    To describe our construction, it will be convenient to define the following two strings, both of length \((n+1)\):
    \begin{align*}
        X_1 &:= (n+1)\circ(n+2)\circ \cdots \circ (2n+1), \\
        X_2 &:= (2n+2)\circ(2n+3)\circ \cdots \circ (3n+2).
    \end{align*}
    Next, we define the two special permutations \(\pi^L, \pi^R \in \mathcal{S}_{N}\) alluded to earlier:
    \begin{align*}
        \pi^L &:= 1\circ 2 \circ \cdots \circ (n-1) \circ n \circ X_1 \circ X_2, \\
        \pi^R &:= X_1 \circ n \circ (n-1) \circ \cdots \circ 2 \circ 1 \circ X_2.
    \end{align*}
We make the observation that every median of the set \(\{\pi^{L}, \pi^R\}\) naturally encodes a cut of the vertex set \(V\).
\begin{defi}
    For a nonnegative integer \(r \leq n\), let \(A = \{a_1, a_2, \ldots , a_r\}\) and \(B=\{b_1, b_2, \ldots , b_{n-r}\}\) be sets such that \(A \sqcup B = [n]\), \(a_1 < a_2 < \cdots  < a_r\) and \(b_1 > b_2 > \cdots > b_{n-r}\). Define \(\pi^{A, B} \in \mathcal{S}_{N}\) as:
    \[\pi^{A, B} := a_1 \circ a_2 \circ \cdots \circ a_r \circ X_1 \circ b_1 \circ b_2 \circ \cdots \circ b_{n-r} \circ X_2.\]
    Furthermore, define \(\mathcal{S}_{N}^{*} \subseteq \mathcal{S}_{N}\) as:
    \[\mathcal{S}_{N}^{*}:= \{\pi \in \mathcal{S}_N: \pi = \pi^{A, B} \textnormal{ for some pair of sets \(A, B\) with \(A \sqcup B = [n]\)} \}.\]
\end{defi}
Clearly, permutations in \(\mathcal{S}_N^*\) naturally encode cuts of \([n]\). We will first show that every permutation in \(\mathcal{S}_N^*\) has the same sum of Ulam distances to \(\pi^L\) and \(\pi^R\).
\begin{lemma}
\label{lemma:legal-cost}
    For every \(\pi \in \mathcal{S}_N^*\), \(d_U(\pi, \pi^L) + d_U(\pi, \pi^R) = n\).
\end{lemma}
\begin{proof}
\renewcommand\qedsymbol{\(\lrcorner\)}
    Fix some \(\pi \in \mathcal{S}_N^*\). Then \(\pi = \pi^{A, B}\) for some sets \(A, B\) with \(A\sqcup B = [n]\). Let \(|A| = r\) and \(|B| = n-r\) for some nonnegative integer \(r\leq n\). Clearly, \(\LCS(\pi, \pi^L)\geq r + 2(n+1)\), since one can form a common subsequence of \(\pi\) and \(\pi^L\) by deleting the symbols in \(B\) from both. Furthermore, no LCS of \(\pi\) and \(\pi^L\) can contain a symbol from \(B\) since otherwise, the LCS would not contain any symbol from \(X_1\) and would have a length that is at most \(2n+1\). Therefore, \(\LCS(\pi, \pi^L) = r + 2(n+1)\) and \(d_U(\pi, \pi^L) = n-r\). By a similar argument, \(d_U(\pi, \pi^R) = r\). So, \(d_U(\pi, \pi^L) + d_U(\pi, \pi^R) = n\).
\end{proof}
Next, we show that any permutation that is not in \(\mathcal{S}_N^*\) has strictly larger sum of Ulam distances to \(\pi^L\) and \(\pi^R\).
\begin{restatable}{lemma}{tedious}
\label{lemma:suboptimality}
    Let \(\pi \notin \mathcal{S}_N^*\). Then \(d_U(\pi, \pi^L) + d_U(\pi, \pi^R) \geq n+1\). 
\end{restatable}
The proof of Lemma~\ref{lemma:suboptimality} involves somewhat tedious casework and has been deferred to Appendix~\ref{sec:suboptimality}.



The final pieces in our construction are the ``edge gadgets'', which we now define. For each edge \(e = \{i, j\} \in E\), where \(i < j\), define the following two strings:
\begin{align*}
    \pi_e^1 &= j\circ i\circ X_1 \circ X_2 \circ 1 \circ 2 \circ \cdots \circ (i-1) \circ (i+1) \circ \cdots \circ (j-1) \circ (j+1) \circ \cdots \circ n, \\
    \pi_e^2 &=  X_1 \circ i\circ j\circ X_2 \circ 1 \circ 2 \circ \cdots \circ (i-1) \circ (i+1) \circ \cdots \circ (j-1) \circ (j+1) \circ \cdots \circ n.
\end{align*}
The role of the edge gadgets associated with an edge \(e\) is to reward choosing partitions of the vertices that cut \(e\). This is formalized in the following lemma.
\begin{lemma}
\label{lemma:edge-gadget}
    Let \(\pi\in \mathcal{S}^*_N\) such that \(\pi = \pi^{A, B}\) with \(A \sqcup B =[n]\) and \(e\in E\). Then we have the following:
    \[d_U(\pi_e^1, \pi ) +d_U(\pi_e^2, \pi)=\begin{cases}
        2n-2, & \text{if \(e\) is cut by the partition \((A, B)\),}\\
        2n-1, & \text{otherwise.}
    \end{cases}\]
\end{lemma}
\begin{proof}
\renewcommand\qedsymbol{\(\lrcorner\)}
    Suppose \(e\) is cut by the partition \((A, B)\) and let \(\{k\}=A\cap e\). We must have \(\LCS(\pi_e^1, \pi) \geq 2n+3\) since the string \(k \circ X_1 \circ X_2\) is a common subsequence of \(\pi_e^1\) and \(\pi\).
    Furthermore, no LCS of \(\pi_e^1\) and \(\pi\) can contain any symbol in \([n]\setminus\{k\}\) since otherwise, the LCS would not contain any symbol from one of \(X_1\) and \(X_2\), and consequently, would have a length of at most \(2n+1\). Therefore, \(\LCS(\pi_e^1, \pi) = 2n+3\) and \(d_U(\pi_e^1, \pi)=n-1\). By essentially the same argument, we also have \(d_U(\pi_e^2, \pi )=n-1\) and therefore, \(d_U(\pi_e^1, \pi ) +d_U(\pi_e^2, \pi)=2n-2\).

    If \(e\) is not cut by the partition \((A, B)\), then assume that \(e=\{i, j\} \subseteq A\). The case where \(e\subseteq B\) is similar. First, note that \(\LCS(\pi_e^1, \pi) \geq 2n+3\) since the string \(i \circ X_1 \circ X_2\) is a common subsequence of \(\pi_e^1\) and \(\pi\). Furthermore, \(\LCS(\pi_e^1, \pi) \leq 2n+3\) because of the following two reasons:
    \begin{itemize}
        \item No LCS of \(\pi_e^1\) and \(\pi\) can contain both \(i\) and \(j\) since their relative order is different in \(\pi_e^1\) and \(\pi\).
        \item No LCS of \(\pi_e^1\) and \(\pi\) can contain any symbol in \([n]\setminus\{i, j\}\) since otherwise, the LCS would not contain any symbol from one of \(X_1\) and \(X_2\), and consequently, would have a length of at most \(2n+1\).
    \end{itemize}
    Therefore, \(\LCS(\pi_e^1, \pi) = 2n+3\) and \(d_U(\pi_e^1, \pi)=n-1\).
    Next we claim that \(\LCS(\pi_e^2, \pi) \leq 2n+2\). Let \(\rho\) be any LCS of \(\pi_e^2\) and \(\pi\). We make the following observations.
    \begin{itemize}
        \item If \(\rho\) contains any symbol from \(A\), then \(\rho\) cannot contain any symbol from \(X_1\) forcing \(|\rho|\leq 2n+1\).
        \item If \(\rho\) contains any symbol from \(B\), then \(\rho\) cannot contain any symbol from \(X_2\), once again forcing \(|\rho|\leq 2n+1\). This is because \(B\) contains neither \(i\) nor \(j\).
    \end{itemize}
    Therefore, the symbols of \(\rho\) must come from \(X_1\) or \(X_2\) and so, \(|\rho| \leq |X_1| +|X_2| =2n+2\). Furthermore, \(X_1\circ X_2\) is a common subsequence of \(\pi_e^2\) and \(\pi\). Therefore, \(\LCS(\pi_e^1, \pi) = 2n+2\) and \(d_U(\pi_e^1, \pi)=n\). Putting everything together, we have \(d_U(\pi_e^1, \pi ) +d_U(\pi_e^2, \pi)=2n-1\).
\end{proof}
Now let \(S_E=\{\pi_e^1 : e\in E\} \cup \{\pi_e^2 : e\in E\}\) and \(S_{\text{aux}}\) be the set consisting of \(t:=|E|(2n-1)\) copies of \(\pi^L\) and \(\pi^R\). Our final \cm instance will be the multiset
\[S:= S_E \cup S_{\text{aux}}.\]
We now show that  \(G\) has a cut of size at least \(k\) if and only if the median of \(S\) has cost at most \(k'\), where \(k'= |E|(2n-1) - k +tn\). For the completeness case, assume there exist sets \(A, B\) with \(A\sqcup B =n\) such that the partition \((A, B)\) cuts at least \(k\) edges in \(G\). Denote by \(E(A, B)\) the set of all edges cut by the partition \((A, B)\). Now consider the permutation \(\pi^{A, B} \in \mathcal{S}_n\) and note that:

\begin{small}
\begin{align*}
    \sum_{\pi\in S} d_U(\pi^{A, B}, \pi) & = \sum_{\pi\in S_E}d_U(\pi^{A, B}, \pi) + \sum_{\pi\in S_{\text{aux}}} d_U(\pi^{A, B}, \pi)\\
    & = \left(\sum_{e\in E}(d_U(\pi^{A, B}, \pi_e^1)+d_U(\pi^{A, B}, \pi_e^2))\right) + t(d_U(\pi^{A, B}, \pi^L)+d_U(\pi^{A, B}, \pi^R))\\
    &=\left(\sum_{e\in E(A, B)}(d_U(\pi^{A, B}, \pi_e^1)+d_U(\pi^{A, B}, \pi_e^2)) + \sum_{e\notin E(A, B)}(d_U(\pi^{A, B}, \pi_e^1)+d_U(\pi^{A, B}, \pi_e^2))\right) + tn\\
    & =\left(|E(A, B)|(2n-2)+ (|E|-|E(A, B)|)(2n-1)\right) + tn\\
    & = |E|(2n-1) - |E(A, B)| +tn\\
    & \leq |E|(2n-1) - k +tn.
\end{align*}
\end{small}

For the soundness case, assume that there exists \(\pi^* \in \mathcal{S}_N\) whose median cost to \(S\) is at most \(k'\). We can further assume that \(\pi^*\in \mathcal{S}^*_N\) since otherwise, every \(\tilde{\pi}\in \mathcal{S}^*_N\) will have a median cost that is at most that of \(\pi^*\). Indeed, assume that \(\pi^*\notin \mathcal{S}^*_N\) and fix any \(\tilde{\pi}\in \mathcal{S}^*_N\). We have:
\begin{align*}
    \sum_{\pi\in S}d_U(\pi^*, \pi) &=\sum_{\pi\in S_{\text{aux}}}d_U(\pi^*, \pi) +\sum_{\pi\in S_E}d_U(\pi^*, \pi)\\
    &\geq t(n+1) + 0\\
    &= tn +t\\
    &=tn + |E|(2n-1)\\
    &\geq tn +\sum_{\pi\in S_{E}}d_U(\tilde{\pi}, \pi)\\
    & = \sum_{\pi\in S_{\text{aux}}}d_U(\tilde\pi, \pi) +\sum_{\pi\in S_E}d_U(\tilde\pi, \pi)= \sum_{\pi\in S}d_U(\tilde\pi, \pi).
\end{align*}
Thus, the assumption that \(\pi^*\in \mathcal{S}^*_N\) is without loss of generality. Then, we have \(\pi^*=\pi^{A, B}\) for sets \(A, B\) with \(A\sqcup B =[n]\). We claim that the partition \((A, B)\) cuts at least \(k\) edges in \(G\). Indeed, since \(\sum_{\pi\in S_{\text{aux}}}d_U(\pi^*, \pi)=tn\), we have \(\sum_{\pi\in S_{E}}d_U(\pi^*, \pi) \leq k' - tn = |E|(2n-1)-k = k(2n-2)+(|E|-k)(2n-1)\). Then, by Lemma~\ref{lemma:edge-gadget} the partition \((A, B)\) cuts at least \(k\) edges.
\end{proof}

\begin{remark}
    Although our NP-hardness reduction produces multisets instead of sets, it is not to hard to turn them into sets by appending a unique permutation to each permutation without affecting the structure of the solution. See Appendix~\ref{sec:from-multisets-to-sets} for details. 
\end{remark}


\section{Fine-Grained Complexity of Discrete Center in the Ulam Metric}
\label{sec:discrete-center}

In this section, we prove Theorem~\ref{thm:result-3}. We first formally define the discrete center problem in the Ulam metric.

\begin{table}[h]
    \centering
    \renewcommand{\arraystretch}{1.4} 

    \begin{tabular}{|p{0.927\textwidth}|}
        \hline
        \dc \\
        \hline
    \end{tabular}

    \begin{tabular}{|p{0.15\textwidth}|p{0.75\textwidth}|}
        \hline
        \textbf{Input:} & A set \(S\subseteq \mathcal{S}_L\) of permutations such that \(|S|=n\) and an integer \(\tau\).\\
        \hline
        \textbf{Question:} & Is there a permutation \(\pi^*\in S\) such that \(\max_{\pi\in S}d_U(\pi, \pi^*)\leq \tau\)?  \\
        \hline
    \end{tabular}

\end{table}

\noindent Our main result of this section is the following.

\begin{theorem}
\label{thm:discrete-center}
Let $\varepsilon > 0$ and $\alpha > 0$. There is no algorithm running in time $O((n^2L)^{1-\varepsilon})$ that solves the \dc problem for $n$ permutations of length $L = \Theta(n^{\alpha})$, unless \textnormal{\eaeSETH} fails.
\end{theorem}


The key step in the proof of Theorem~\ref{thm:discrete-center} is a reduction from Orthogonal Vectors to Ulam Distance -- that is, to construct a pair of functions that, given a set of \(n\) binary vectors of length \(d\) each as input, outputs, independently of each other, a pair of permutations whose Ulam distance is small if and only if there exists an orthogonal pair of vectors in the input sets. 

\begin{theorem}
\label{thm:ov-to-ulam}
    There exists a pair of functions \(f\) and \(g\) such that for all sets \(A, B \subseteq \{0, 1\}^d\) with \(\lvert A \rvert = \lvert B \rvert = n\), the following holds.
    \begin{itemize}
        \item \(f(A), g(B) \in \mathcal{S}_{(5d-1)n^2}\), i.e., both \(f\) and \(g\) output permutations of length \((5d-1)n^2\).
        \item If there exist \(a\in A, b\in B\) such that \(\langle a, b\rangle = 0\), then the Ulam distance between \(f(A)\) and \(g(B)\) is at most \(3n^2d-1\). Otherwise, the Ulam distance between \(f(A)\) and \(g(B)\) is exactly \(3n^2d\).
        \item Both \(f\) and \(g\) are computable in time \(O(n^2d)\).
    \end{itemize}
\end{theorem}

\begin{proof}
    The proof outline resembles many of the other sequence dissimilarity (e.g., edit distance) lower bound proofs --  we will first construct coordinate and vector gadgets, and then stitch them together to construct the final output permutations (see~\cite{bringmann2019lectures} for an instructive example). Throughout the proof, we will phrase our arguments in terms of LCS lengths and use Fact~\ref{fact:ulam-lcs} at the end. We begin by constructing coordinate gadgets, which are small permutations that encode the product of two bits through their LCS.
    \begin{lemma}[Coordinate Gadgets]
    There exist functions \(C_A, C_B: \{0, 1\}\to \mathcal{S}_3\) such that for all \(x, y\in \{0, 1\}\), \(\LCS(C_A(x), C_B(y)) = 2-xy\).
    \end{lemma}
    \begin{proof}
    \renewcommand\qedsymbol{\(\lrcorner\)}
        For each \(x, y\in \{0, 1\}\), let \(C_A(x)\) and \(C_B(y)\) be defined as follows:
        \begin{alignat*}{2}
            C_A(0) &:= 123, \qquad\qquad & C_A(1) &:= 312, \\
            C_B(0) &:= 132, \qquad\qquad & C_B(1) &:= 213.
        \end{alignat*}
        By inspection, \(\LCS(C_A(0), C_B(0)) = \LCS(C_A(0), C_B(1)) = \LCS(C_A(1), C_B(0)) = 2\), Furthermore, \(\LCS(C_A(1), C_B(1)) =1\).
    \end{proof}
    Next, from these coordinate gadgets, we construct vector gadgets that guarantee a large LCS if and only if the corresponding vectors have a small inner product.
    \begin{lemma}[Vector Gadgets]
    \label{lemma:vec-gadget}
        There exist functions \(V_A, V_B: \{0, 1\}^d \to \mathcal{S}_{3d}\), computable in time \(O(d)\), such that for all vectors \(a, b\in \{0, 1\}^d\), we have
        \(\LCS(V_A(a), V_B(b)) = 2d - \langle a, b \rangle\).
    \end{lemma}
    \begin{proof}
    \renewcommand\qedsymbol{\(\lrcorner\)}
        Fix \(a, b \in \{0, 1\}^d\) where \(a= a[1]a[2]\ldots a[d]\) and \(b = b[1]b[2]\ldots b[d]\). We will construct the permutations \(V_A(a)\) and \(V_B(b)\) by concatenating the coordinate gadgets for \(a[i]\) and \(b[i]\), respectively, for each \(i\in [d]\), while using fresh new symbols to ensure that no letters are repeated. More precisely, for each \(i\in [d]\), we define\footnote{Recall that for a string \(s\) and a nonnegative integer \(k\), we denote by \(\Delta_k(s)\) the string obtained by adding \(k\) to every symbol of \(s\). } \(c_{a, i} := \Delta_{3i-3}(C_A(a[i]))\) and \(c_{b, i} := \Delta_{3i-3}(C_B(b[i]))\). Finally, we define \(V_A(a)\) and \(V_B(b)\) as follows:
        \begin{align*}
            V_A(a) &:=\bigcircop_{i \in [d]}c_{a, i}, \\
            V_B(b) &:=\bigcircop_{i\in [d]}c_{b, i}.
        \end{align*}
        Clearly, for distinct \(i, j \in [d]\), \(c_{a, i}\) and \(c_{b, j}\) have no symbols in common. So, for every \(i\in [d]\), the symbols of \(c_{a, i}\) can only be aligned with symbols of \(c_{b, i}\), and we have
        \begin{align*}
        \LCS(V_A(a), V_B(b)) & = \sum_{i\in[d]} \LCS\left(c_{a, i}, c_{b, i}\right) \\
        & =\sum_{i\in [d]} \left(2-a[i]b[i]\right) \\
        & = 2d - \langle a, b\rangle. \qedhere
        \end{align*}
        \end{proof}
    We then \textit{normalize} these vector gadgets to ensure only two possible values for the LCS length of two gadgets.
    \begin{lemma}[Normalized Vector Gadgets]
    \label{lemma:norm-vec-gadget}
        There exist functions \(N_A, N_B: \{0, 1\}^d \to \mathcal{S}_{5d-1}\), computable in time \(O(d)\), such that for all vectors \(a, b \in \{0, 1\}^d\), we have
        \[\LCS(N_A(a), N_B(b)) = \begin{cases} 
        2d, & \text{if } \langle a, b \rangle = 0, \\
        2d-1, & \text{if } \langle a, b \rangle \neq 0.
        \end{cases}\]
    \end{lemma}
    \begin{proof}
    \renewcommand\qedsymbol{\(\lrcorner\)}
        Fix \(a, b \in \{0, 1\}^d\). Let \(P:= \bigcircop_{i\in [2d-1]} (3d+i)\). We define
        \begin{align*}
            N_A(a) &:= P\circ V_A(a), \\
            N_B(b) &:= V_B(b)\circ P.
        \end{align*}
        To see why the claim is true, note that since \(P\) has no symbols in common with \(V_A(a)\) or \(V_B(b)\), \(\LCS(N_A(a), N_B(b)) = \max\{\lvert P \rvert, \LCS(V_A(a), V_B(b))\}\). Moreover, by Lemma~\ref{lemma:vec-gadget}, if \(\langle a, b \rangle =0\), \( \LCS(V_A(a), V_B(b)) = 2d > 2d -1 = \lvert P \rvert\). In this case, \(\LCS(N_A(a), N_B(b)) = \LCS(V_A(a), V_B(b)) = 2d\). In the case where \(\langle a, b \rangle \neq 0\), we have \(\LCS(V_A(a), V_B(b)) \leq 2d -1 = \lvert P \rvert\). We finally conclude that \(\LCS(N_A(a), N_B(b)) = \lvert P \rvert = 2d-1\).
    \end{proof}

    \begin{figure}
        \centering
        \includegraphics[scale=0.39]{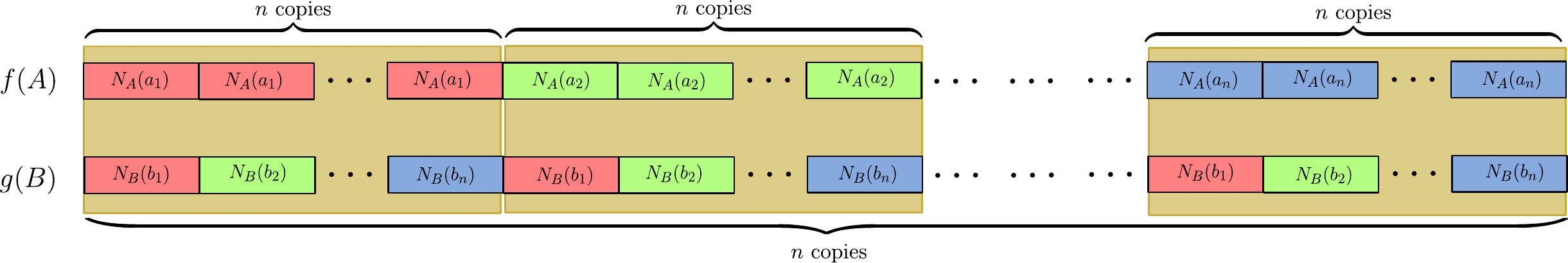}
        \caption{Illustration of the functions \(f\) and \(g\). \(f(A)\) is obtained by taking \(n\) copies of each normalized vector gadget in \(A\) and then concatenating together the \(n\) resulting strings. For \(g(B)\), the order is flipped -- the \(n\) normalized vector gadgets in \(B\) are first concatenated together and then~\(n\) copies of the resulting string are concatenated.}
        \label{fig:or-gadget}
    \end{figure}
Given these normalized vector gadgets, we are now ready to describe how to compute the functions \(f\) and \(g\). Fix \(A, B \subseteq \{0, 1\}^d\) with \(\lvert A \rvert = \lvert B \rvert = n\), and let \(A =\{a_1, a_2, \ldots , a_n\}\) and \(B =\{b_1, b_2, \ldots , b_n\}\). On a high-level, we obtain the permutations \(f(A)\) and \(g(B)\) by concatenating the normalized vector gadgets for each vector in \(A\) and \(B\), respectively, some number of times and in some particular order. For \(f(A)\), the idea is to take \(n\) copies of each normalized vector gadget in \(A\) and then concatenating together the \(n\) resulting strings, while using fresh new symbols when necessary. For \(g(B)\), we flip the order -- we first concatenate together the \(n\) normalized vector gadgets in \(B\) and then take \(n\) copies of the resulting string, again using fresh new symbols when necessary. See Figure~\ref{fig:or-gadget} for an illustration. More precisely, we do the following. Let \(v = 5d-1\) be the length of each normalized vector gadget. For each \((i, j)\in [n]\times [n]\), we define:
\begin{align*}
    \alpha_{i, j} &:= \Delta_{((i-1)n+(j-1))v}(N_A(a_i)), \\
    \beta_{i, j} &:= \Delta_{((i-1)n+(j-1))v}(N_B(b_j)).
\end{align*}
Finally, we define
\begin{align*}
    f(A) &:=\bigcircop_{i\in [n]}\left(\bigcircop_{j \in [n]} \alpha_{i, j}\right), \\
    g(B) &:=\bigcircop_{i\in [n]}\left(\bigcircop_{j \in [n]} \beta_{i, j}\right).
\end{align*}

Clearly, \(\lvert f(A)\rvert = \lvert g(B)\rvert = vn^2 = (5d-1)n^2\). So, it only remains to prove the Ulam distance claim in Theorem~\ref{thm:ov-to-ulam}. Note that for distinct pairs \((i, j), (i', j') \in [n]\times [n]\), \(\alpha_{i, j}\) and \(\beta_{i', j'}\) do not have any symbols in common. Therefore, for every \((i, j) \in [n]\times [n]\), the symbols of \(\alpha_{i, j}\) can only be aligned with the symbols of \(\beta_{i, j}\), and we have
\[\LCS(f(A), g(B)) = \sum_{(i, j)\in [n]\times [n]} \LCS(\alpha_{i, j}, \beta_{i, j}).\]
Furthermore, by Lemma~\ref{lemma:norm-vec-gadget}, for each pair \((i, j)\), \(\LCS(\alpha_{i, j}, \beta_{i, j})\) can take only one of two values -- \(2d\), which happens when \(\langle a_i, b_j\rangle = 0\), or \(2d-1\), which happens when \(\langle a_i, b_j\rangle \neq 0\). Consider the set \(S = \{(i, j)\in [n]\times [n]: \langle a_i, b_j \rangle =0\}\). We now have:
\begin{align*}
    \LCS(f(A), g(B)) & = \sum_{(i, j)\in [n]\times [n]} \LCS(\alpha_{i, j}, \beta_{i, j}) \\
    & = \sum_{(i, j) \in S} \LCS(\alpha_{i, j}, \beta_{i, j}) + \sum_{(i, j) \notin S} \LCS(\alpha_{i, j}, \beta_{i, j}) \\
    & = 2d \lvert S \rvert + (2d-1)(n^2-\lvert S \rvert).
\end{align*}
This means that if there do not exist \(a_i\in A, b_j \in B\) such that \(\langle a_i, b_j \rangle =0\) (i.e., \(S = \emptyset\)), then \(\LCS(f(A), g(B)) = (2d-1)n^2\). Otherwise, \(\LCS(f(A), g(B)) \geq (2d-1)n^2 +1\). The claim then follows by Fact~\ref{fact:ulam-lcs}.
\end{proof}
Equipped with Theorem~\ref{thm:ov-to-ulam}, we can now prove Theorem~\ref{thm:discrete-center}.
\begin{proof}[Proof of Theorem~\ref{thm:discrete-center}]
    We will give a reduction from \(\exists\forall\exists\exists\)OV to the \dc problem. The reduction will be in two steps. In the first step, we will reduce to a problem called \textsc{Bichromatic Discrete Ulam Center}, where one is given as input \textit{two} sets \(X\) and \(Y\) of permutations along with an integer \(\tau\), and the goal is to determine if there is a permutation in \(X\) that has Ulam distance at most \(\tau\) to all permutations in \(Y\). In the second step, we will reduce \textsc{Bichromatic Discrete Ulam Center} to \dc.

     Let \(A, B, C, E \subseteq\{0, 1\}^{d}\) be any \(\exists\forall\forall\exists\exists\)OV instance such that \(|A|= |B| = n\) and \(|C|=|E|=m\), where \(m=n^{\Theta(1)}\) and \(\omega(\log n) < d < n^{o(1)}\). From this, we construct a pair of sets \(X\) and \(Y\) of permutation as follows.

    For each \(a\in A\), we first construct the set \(V_a\) consisting of \(m\) vectors as follows: \[V_{a} =\{a\odot c\in \{0, 1\}^d: c\in C\}.\] 
    Similarly, for each \(b\in B\), we construct the set \(W_{b}\), again consisting of \(m\) vectors as \[W_{b} =\{b\odot e \in \{0, 1\}^d: e\in E\}.\]
    
    Finally, for each set \(V_{a}\), where \(a\in A\), we compute \(f(V_a)\) (where \(f\) is one of the functions from Theorem~\ref{thm:ov-to-ulam}) and add the resulting permutation to the first set \(X\). That is we consider:
    \[X =\{f(V_a) : a \in A\}.\]
    Similarly, construct
    \[Y =\{g(W_{b}) : b \in B\}.\]
    
    Finally, we set \(\tau = 3m^2d-1\). We now have two sets \(X\) and \(Y\) each containing \(n\) permutations of length \((5d-1)m^2 = L\). Furthermore, by Theorem~\ref{thm:ov-to-ulam}, there exists \(x\in X\)  such that \(d_U(x, y)\leq \tau\) for each \(y\in Y\) if and only if there exists an orthogonal pair in \(V_a,W_b\) for some \(a\in A\) and all \(b\in B\) which happens if and only if the original \(\exists\forall\exists\exists\text{OV}\) instance is a YES-instance. 

    The entire reduction takes \(O\left(nm^2d\right) = O(nL)\) time. Now fix \(\varepsilon >0\). If the resulting \textsc{Bichromatic Discrete Ulam Center} instance can be solved in time \(O((n^2L)^{1-\varepsilon})\), then \(\exists\forall\forall\exists\exists \text{OV}\) can also be solved in time \(O((n^2m^2d)^{1-\varepsilon}) = O((n^{2+o(1)}m^2)^{1-\varepsilon}) =O((|A||B||C||E|)^{1-\varepsilon'})\) for some constant \(\varepsilon'>0\), refuting \ueaeeOVH. 

    Next we reduce \textsc{Bichromatic Discrete Ulam Center} to \dc. We first define the strings \(p_1 := \bigcircop_{i\in [L]} (L+i)\), \(p_2 := \bigcircop_{i \in [2L]} (2L+i)\), and \(p_3 := \bigcircop_{i\in [L]} i\). Furthermore, we define \(p_1^{\mathcal{R}}\) and \(p_2^{\mathcal{R}}\) to be the reverses of \(p_1\) and \(p_2\), respectively. Our \dc instance \(S \subseteq \mathcal{S}_{4m}\) is simply the following:
    \[S =\{p_2^{\mathcal{R}} \circ p_1 \circ x : x \in X\} \cup \{p_2 \circ p_1 \circ y : y\in Y\} \cup \{p_2^{\mathcal{R}} \circ p_1^{\mathcal{R}} \circ p_3\}.\]
    Note that the center of \(S\) must be of the form \(p_2^{\mathcal{R}} \circ p_1 \circ x\), where \(x\in X\). This is seen from the following observation, which can be justified using straightforward alignment arguments.
    \begin{obs}
    \label{obs: fac-elim}
        Let \(x, y \in \mathcal{S}_L\). Then the following hold.
        \begin{itemize}
            \item \(d_U(p_2^{\mathcal{R}} \circ p_1 \circ x, p_2^{\mathcal{R}} \circ p_1 \circ y) \leq L-1\).
            \item \(d_U(p_2^{\mathcal{R}} \circ p_1 \circ x, p_2^{\mathcal{R}} \circ p_1^{\mathcal{R}} \circ p_3) \leq 2L-2\).
            \item \(d_U(p_2 \circ p_1 \circ y, p_2^{\mathcal{R}} \circ p_1^{\mathcal{R}} \circ p_3) \geq 3L-2\).
            \item \(d_U(p_2^{\mathcal{R}} \circ p_1 \circ x, p_2 \circ p_1 \circ y) = 2L -1+ d_U(x, y)\).
        \end{itemize}
    \end{obs}
    By Observation~\ref{obs: fac-elim}, \(S\) has a center with cost at most \(2L-1+\tau\) if and only if the \textsc{Bichromatic Discrete Ulam Center} instance \(X \cup Y\) has a solution with cost at most \(\tau\).  
\end{proof}

\subsection{The Need for Quantifiers}
\label{sec:quant-justification}

Our lower bound for \dc is based on a plausible generalization of the Strong Exponential Time Hypothesis, namely \eaeSETH. One could ask if we could get a similar lower bound under the weaker but more standard \SETH instead. We remark that this is impossible unless the Nondeterministic Strong Exponential Time Hypothesis (\textsf{NSETH})~\cite{CarmosinoGIMPS16} is false. This is because \dc can be solved in (co-)nondeterministic time $\tilde O(n L)$ with the following simple algorithm. We first guess the center $\pi_C$ and compute~\smash{$d^* := \max_{\pi \in X} d_U(\pi_C, \pi)$}. Then for each permutation $\pi$, we guess the furthest permutation $\pi' \in X$ and verify that $d_U(\pi, \pi') \geq d^*$, thereby certifying that our guess $\pi_C$ is optimal. In light of this algorithm, a~\smash{$O((n^2 L)^{1-\Omega(1)})$} fine-grained lower bound based on \SETH would contradict \textsf{NSETH}.

\section{Fine-Grained Complexity of Discrete Median in the Ulam Metric} \label{sec:median}

In this section, we prove a tight fine-grained lower bound for the \textsc{Discrete Ulam Median} problem conditioned on \SETH. We first give a formal statement of the problem.

\begin{table}[h]
    \centering
    \renewcommand{\arraystretch}{1.4} 

    \begin{tabular}{|p{0.927\textwidth}|}
        \hline
        \dm \\
        \hline
    \end{tabular}

    \begin{tabular}{|p{0.15\textwidth}|p{0.75\textwidth}|}
        \hline
        \textbf{Input:} & A set \(S\subseteq \mathcal{S}_L\) of permutations such that \(|S|=n\) and an integer \(\tau\).\\
        \hline
        \textbf{Question:} & Is there a permutation \(\pi^*\in S\) such that \(\sum_{\pi\in S}d_U(\pi, \pi^*)\leq \tau\)?  \\
        \hline
    \end{tabular}

\end{table}

This section is organized as follows. In Section~\ref{sec:median:sec:bichromatic}, we start with a simple lower bound for the \emph{bichromatic} version of the problem that essentially follows from the proof for the \textsc{Discrete Ulam Center} problem from before. Going from the bichromatic version to the standard monochromatic version defined above is technically significantly more involved. We establish this stronger lower bound in three steps. First, in Section~\ref{sec:median:sec:embedding}, we prepare some useful lemmas for embedding the Hamming metric (in specialized settings) to the Ulam metric. This will be particularly helpful in the design of some gadgets. Second, in Section~\ref{sec:median:sec:balancing}, we establish the core of our reduction and show that any set of $n$ permutations can be \emph{balanced} in such a way that the sum of Ulam distances from each permutation to all the others is (approximately) \emph{equal}. Finally, in Section~\ref{sec:median:sec:monochromatic}, we assemble the full reduction.

\subsection{Hardness for Bichromatic Instances} \label{sec:median:sec:bichromatic}

In this section, we give a fine-grained lower bound for the \textsc{Bichromatic} \dm problem. In this problem, we given \textit{two} sets \(X, Y\) of permutations and an integer \(\tau\), and the goal is to determine if there exists \(x\in X\) such that \(\sum_{y\in Y}d_U(x, y) \leq \tau\). 

\begin{theorem} \label{thm:bichromatic-ulam-median}
    Let $\varepsilon > 0$ and $\alpha > 0$. There is no algorithm running in time $O((n^2L)^{1-\varepsilon})$ that solves the \textsc{Bichromatic} \dm problem for $n$ permutations of length $L = \Theta(n^{\alpha})$, unless the SETH fails.
\end{theorem}
\begin{proof}
    The proof is almost identical to the first half of the proof of Theorem~\ref{thm:discrete-center}. The only difference is the starting problem, which is 4-\text{OV} (i.e., \(\exists\exists\exists\exists\text{OV}\)) instead of  \(\exists\forall\exists\exists\text{OV}\). Given a 4-OV instance \(A, B, C, E\subseteq \{0, 1\}^d\) where \(|A|=|B|=n\), \(|C|=|E|=m = n^{\Theta(1)}\) and \(\omega(\log n) < d < n^{o(1)}\),  we retrace the proof of Theorem~\ref{thm:discrete-center} and construct the sets \(X\) and \(Y\). Finally, we set \(\tau := 3m^2nd-1\). Once again, by Theorem~\ref{thm:ov-to-ulam}, it is not hard to see that there exists \(x\in X\) such that \(\sum_{y\in Y}d_U(x, y)\leq \tau\) if and only if the starting 4-OV instance is a YES-instance. The conclusion then follows from \fourOVH.
\end{proof}

\subsection{Embedding the Hamming Metric into the Ulam Metric} \label{sec:median:sec:embedding}

The following lemmas show how to embed the Hamming metric on various sets of strings into the Ulam metric. These lemmas will be useful when we finally reduce \textsc{Bichromatic} \dm into \dm.

\begin{lemma}[Embedding the Hamming Metric on Small Alphabets] \label{lemma:embedding-hd-small-alphabet}
Let $a_1, \dots, a_n \in \{0, 1, 2\}^L$. In time $O(n L)$ we can construct permutations $\pi_1, \dots, \pi_n \in \mathcal S_{3L}$ such that $d_H(a_i, a_j) = d_U(\pi_i, \pi_j)$ for all~\makebox{$i, j \in [n]$}.
\end{lemma}
\begin{proof}
Consider the three permutations $\sigma_0 = 123$, $\sigma_1 = 231$, $\sigma_2 = 312$ in $\mathcal S_3$. Clearly the Ulam distance between any two of these is exactly $1$. Thus, construct the permutation $\pi_i \in \mathcal S_{3L}$ from~$a_i$ by replacing each symbol ``$0$'' by a copy of $\sigma_0$, each symbol ``$1$'' by a copy of $\sigma_1$, and each symbol~``$2$'' by a copy of $\sigma_2$. Here we use fresh symbols for each copying such that across all strings $a_i$, the $j^{\text{th}}$ character is consistently replaced by the same three symbols. The correctness is straightforward, and clearly the strings can be computed in time $O(n L)$.
\end{proof}

\begin{lemma}[Embedding the Hamming Metric on Permutations] \label{lemma:embedding-hd-perm}
Let $\pi_1, \dots, \pi_n \in \mathcal S_L$. In time $O(n L)$ we can compute permutations $\tau_1, \dots, \tau_n \in \mathcal S_{2L}$ such that $d_H(\pi_i, \pi_j) = d_U(\tau_i, \tau_j)$ for all~\makebox{$i, j \in [n]$}.
\end{lemma}
\begin{proof}

Consider the map \(\eta:\mathcal{S}_L\to\mathcal{S}_{2L}\) defined as follows. For each \(\pi\in \mathcal{S}_L\) where \(\pi = \pi[1]\circ \pi[2]\circ \cdots \circ \pi[L]\), the permutation \(\eta(\pi)\in \mathcal{S}_{2L}\) is equal to:
    \[\pi[1]\circ (L+1) \circ \pi[2] \circ (L+2)\circ\cdots \pi[L]\circ (2L).\]
    In other words, \(\eta(\pi)\) is obtained by interleaving the symbols of \(\pi\) with the symbols of the string \((L+1)\circ (L+2)\circ \cdots \circ (2L)\). We now make the following observation. 

    \begin{obs}
    \label{obs:interleaved}
        For every \(\pi, \pi'\in \mathcal{S}_L\), we have \(d_U(\eta(\pi), \eta(\pi'))=d_H(\pi, \pi')\).
    \end{obs}
    \begin{proof}
         \renewcommand\qedsymbol{\(\lrcorner\)}
         Let \(\sigma := \eta(\pi)\) and \(\sigma' := \eta(\pi')\), respectively. It suffices to show that \(\LCS(\sigma, \sigma')=2L-d_H(\pi, \pi')\). Note that \(\LCS(\sigma, \sigma')\geq 2L-d_H(\pi, \pi')\) since one can obtain a common subsequence of \(\sigma\) and \(\sigma'\) with length \(2L-d_H(\pi, \pi')\) by concatenating, for \(i\in [2L]\), the symbols \(\sigma[i]\) if \(\sigma[i]=\sigma'[i]\). To show that \(\LCS(\sigma, \sigma')\leq 2L-d_H(\pi, \pi')\), consider any longest common subsequence \(s\) of \(\sigma\) and \(\sigma'\). If \(s\) never aligns a symbol \(\sigma[i]\) with the symbol \(\sigma'[j]\) with \(i\neq j\), then clearly, \(|s| \leq 2L-d_H(\pi, \pi')\). So, assume otherwise. Then there exist \(i, j \in [2L]\) with \(i\neq j\) such that \(s\) aligns \(\sigma[i]\) with \(\sigma'[j]\). Clearly, \(i\) and \(j\) are both odd since for even \(i, j\), \(\sigma[i] \neq \sigma'[j]\) if \(i\neq j\). We can further assume without loss of generality that \(i<j\). Then \(\sigma[i+1]\) does not appear in \(s\). We can now replace \(\sigma[i]\) with \(\sigma[i+1]\) in \(s\) to obtain a common subsequence that is no shorter but with one fewer alignment between unequal indices. The claim follows. 
    \end{proof}
    For each \(i\in [n]\), we can now set \(\tau_i=\eta(\pi_i)\). Clearly, this takes \(O(nL)\) time and the conclusion follows.
\end{proof}

\begin{lemma}[Embedding the Hamming Metric without Repeating Symbols] \label{lemma:embedding-hd-no-repeat}
Let $a_1, \dots, a_n \in \Sigma^L$ be strings such that each symbol appears at most once in each string $a_i$. Given \(a_1,\dots,a_n\), in time $O(n (|\Sigma| + n \log n))$ we can construct permutations $\pi_1, \dots, \pi_n \in \mathcal S_{O(|\Sigma| + n \log n)}$ and some integer $K \geq 0$ such that $d_H(a_i, a_j) = d_U(\pi_i, \pi_j) - K$ for all distinct~\makebox{$i, j \in [n]$}.
\end{lemma}
\begin{proof}
The key idea of the proof lies in the following claim:

\begin{claim} \label{clm:embedding-hd-no-repeat}
Let $\Sigma_1, \dots, \Sigma_n \subseteq \Sigma$ be alphabets of size $s = |\Sigma_1| = \dots = |\Sigma_n|$. Given \(\Sigma_1,\dots,\Sigma_n\), in time $O(n (|\Sigma| + n \log n))$ we can compute an alphabet $\Phi$ (disjoint from $\Sigma$) and strings $b_1, \dots, b_n$ such that:
\begin{enumerate}
    \item[(i)] Each string $b_i \in (\Sigma_i \sqcup \Phi)^{s + |\Phi|}$ is a permutation of $\Sigma_i \sqcup \Phi$.
    \item[(ii)] For all distinct $i, j$, the strings $b_i$ and $b_j$ have maximal Hamming distance: $d_H(b_i, b_j) = s + |\Phi|$.
    \item[(iii)] $|\Phi| = O(n \log n)$.
\end{enumerate}
\end{claim}
Suppose for now that Claim~\ref{clm:embedding-hd-no-repeat} holds. Then apply the claim to the sets $\Sigma_1, \dots, \Sigma_n$, where $\Sigma_i \subseteq \Sigma$ is the set of characters \emph{not} appearing in $a_i$. Clearly, $s := |\Sigma_1| = \dots = |\Sigma_n| = |\Sigma| - L$. Then the concatenations with the resulting strings $b_1, \dots, b_n$, $\pi_i := a_i b_i$, are permutations of $\Sigma \sqcup \Phi$. Moreover, we have that $d_H(a_i, a_j) = d_H(\pi_i, \pi_j) - K$ for $K := (s + |\Phi|)$ and for all distinct $i, j$. The lemma statement follows by further embedding the permutations $\pi_i$ using the previous lemma. This completes the proof of the lemma, except for the proof of Claim~\ref{clm:embedding-hd-no-repeat}.
\end{proof}

\begin{proof}[Proof of Claim~\ref{clm:embedding-hd-no-repeat}]
We design a recursive algorithm. Throughout assume that $n$ is a power of~2 (i.e., initially add some dummy sets $\Sigma_i$ to the instance to increase $n$ to the closest power of~2). Moreover, assume that $s \geq n$ (otherwise we initially add $n - s$ dummy elements to all alphabets $\Sigma_i$ increasing the number of fresh symbols by at most $n$). The algorithm has two cases:
\begin{enumerate}
    \item If $s > n$: Greedily select pairwise distinct characters $\sigma_1 \in \Sigma_1, \dots, \sigma_n \in \Sigma_n$ (i.e., pick arbitrary symbols $\sigma_i \in \Sigma_i$ different from the previously chosen symbols $\sigma_1, \dots, \sigma_{i-1}$; the condition $|\Sigma_i| = s > n$ ensures there is always one such symbol). Remove these characters from the respective sets (i.e., remove $\sigma_i$ from $\Sigma_i$) and recursively construct strings $b_1', \dots, b_n'$. Then pick the strings $b_i := b_i \sigma_i$.
    \item If $s = n$: Let $k = n/2$. We recursively call the algorithm on $\Sigma_1, \dots, \Sigma_k$ to construct an alphabet $\Phi_1$ and length-$(s + |\Phi_1|)$ strings $b_1', \dots, b_k'$. Similarly, recursively call the algorithm on $\Sigma_{k+1}, \dots, \Sigma_n$ to construct an alphabet $\Phi_2$ and length-$(s + |\Phi_2|)$ strings $b_{k+1}', \dots, b_n'$. As we will see, the sizes of the alphabets $\Phi_1$ and $\Phi_2$ depend only on the number of inputs (which in both cases is $k = n/2$). Therefore, we may identify $\Phi_1$ and $\Phi_2$ (arbitrarily). Let $\Psi$ denote a fresh alphabet of size $|\Psi| = s + |\Phi_1|$. Then let $c_1, \dots, c_k$ denote permutations of $\Psi$ with pairwise Hamming distances $d_H(c_i, c_j) = |\Psi|$ (e.g., let $c_1$ be arbitrary and let $c_2, \dots, c_k$ denote distinct cyclic rotations of $c_1$; as $|\Psi| \geq s \geq n \geq k$ there are sufficiently many such rotations). Finally, let $\Phi := \Phi_1 \sqcup \Psi$ and construct the strings
    \begin{alignat*}{2}
        b_i &:= b_i' \, c_i \qquad &&(i \leq k), \\
        b_i &:= c_{i-k} \, b_i' \qquad &&(i > k).
    \end{alignat*}
\end{enumerate}
This completes the description of the algorithm. It remains to analyze the correctness:
\begin{enumerate}
    \item[(i)] In case 1, we recursively construct strings $b_1', \dots, b_n'$ such that $b_i'$ is a permutation of $(\Sigma_i \setminus \{\sigma_i\}) \sqcup \Phi$. Then $b_i$ is obtained from $b_i'$ by appending~$\sigma_i$, proving the claim. In case 2, each recursively constructed string $b_i'$ is a permutation of $\Sigma_i \sqcup \Phi_1$. Thus, by construction each string $b_i$ is a permutation of $\Sigma_i \sqcup \Phi_1 \sqcup \Psi = \Sigma_i \sqcup \Phi$. 
    \item[(ii)] For case 1, it suffices to observe that we append distinct symbols $\sigma_1, \dots, \sigma_n$ to all recursively constructed strings. So focus on case 2. By induction we get that the Hamming distance between any two strings in $b_1', \dots, b_k'$ is maximal, and similarly that the Hamming distance between any two strings in $b_{k+1}', \dots, b_n'$ is maximal. From this, and by the construction of the strings $c_1, \dots, c_k$, it follows immediately that the Hamming distance between any two strings in $b_1, \dots, b_k$ or any two strings in $b_{k+1}, \dots, b_n$ is maximal. It remains to verify that the Hamming distance $d_H(b_i, b_j)$ is maximal whenever $i \leq k$ and $j > k$:
    \begin{equation*}
        d_H(b_i, b_j) = d_H(b_i' c_i, c_{j-k} b_j') = d_H(b_i', c_{j-k}) + d_H(c_i, b_j') = (s + |\Phi_1|) + (s + |\Phi_2|) = s + |\Phi|.
    \end{equation*}
    Here, we used that all strings $b_i'$ and $c_i$ have the same length, and the fact that the alphabet of the strings $b_i'$ is disjoint from the alphabet of the strings $c_i$.
    \item[(iii)] Let $\phi(n, s)$ denote the size of the constructed alphabet $\Phi$ for the input parameters $n$ and $s$. Then we obtain the following recurrences relations from cases 1 and 2, respectively:
    \begin{alignat*}{2}
        \phi(n, s) &= \phi(n, s-1) \qquad &(s > n), \\
        \phi(n, n) &= 2 \cdot \phi(n/2, n) + n.\qquad
    \end{alignat*}
    Together with the trivial base case $\phi(1, 1) = 0$, this recurrence is solved by $\phi(n, s) = n \log n$ for all $s \geq n$. The running time can be similarly analyzed. \qedhere
\end{enumerate}
\end{proof}

\subsection{Balancing Median Distances} \label{sec:median:sec:balancing}

In this section, we present several lemmas useful for \textit{balancing} median distances. The problem we consider is the following. We are given permutations \(\pi_1, \pi_2, \ldots , \pi_n\) along with the distance sums \(\sum_{j\neq i}d_U(\pi_i, \pi_j)\) for all \(i\in [n]\). Is it possible to append to each \(\pi_i\) some symbols and obtain new permutations \(\pi'_i\) such that the new median distances \(\sum_{j\neq i}d_U(\pi'_i, \pi'_j)\) are all at most two away from each other? The goal would be to do it efficiently while not blowing up the lengths of the permutations by too much. The following lemmas show how to achieve this.

\begin{lemma}[Coarse-Grained Balancing] \label{lemma:balancing-coarse}
Let $n,N$ be integers such that $n$ is divisible by $4$, and let $k_1, \dots, k_n \in [N]$. Given \(k_1,\dots,k_n\), in time $O(N \log n)$ we can construct permutations $\pi_1, \dots, \pi_n, \tau \in \mathcal S_{O(N / n)}$ and integer $d$ with the following two properties:
\begin{itemize}
    \item Writing $d_i = \sum_{j \neq i} d_U(\pi_i, \pi_j)$, it holds that $|(k_i + d_i) - d| \leq n$ for all $i \in [n]$.
    \item It holds that $d_U(\pi_1, \tau) = \dots = d_U(\pi_n, \tau)$.
\end{itemize}
\end{lemma}
\begin{proof}
Instead of constructing the permutations directly, we will construct strings $a_1, \dots, a_n \in \{0, 1\}^{O(N/n)}$ and an integer $d$ such that, writing $d_i = \sum_{j \neq i} d_H(a_i, a_j)$, it holds that $|(d_i + k_i) - d| \leq n$. Then letting $\pi_1, \dots, \pi_n$ denote the images of the strings $a_1, \dots, a_n$ under the embedding from Lemma~\ref{lemma:embedding-hd-small-alphabet}, and letting $\tau$ denote embedding of the all-$2$ string, the lemma statement follows.

The strings $a_1, \dots, a_n$ are constructed by the following process. Initially, let $a_1, \dots, a_n$ be empty strings. Then repeat the following $O(N / n)$ times: Reorder the indices such that $k_1 + d_1 \leq \dots \leq k_n + d_n$. Then:
\begin{itemize}
    \item Append $01$ to the first $\frac{n}{4}$ strings $a_1, \dots, a_{\frac{n}{4}}$.
    \item Append $10$ to the next $\frac{n}{4}$ strings $a_{\frac{n}{4}}, \dots, a_{\frac{n}{2}}$.
    \item Append $00$ to the last $\frac{n}{2}$ strings $a_{\frac{n}{2}+1}, \dots, a_n$.
\end{itemize}

To analyze the correctness of this algorithm, write $\bar d = \frac{1}{n} \sum_i (k_i + d_i)$ and $\Delta_i = (k_i + d_i) - \bar d$. We analyze how the quantities $\Delta_i$ change in a single iteration of the algorithm. For each of the last~$\frac{n}{2}$ strings, $d_i$ increases by $\frac{n}{2}$ (as $d_H(00, 01) = d_H(00, 10) = 1$). For each of the first $\frac{n}{4}$ strings, $d_i$ increases by $2 \cdot \frac{n}{4} + \frac{n}{2} = n$. For each of the middle $\frac{n}{4}$ strings, $d_i$ similarly increases by $n$. Thus, the average $\bar d$ increase is $\frac{3n}{4}$, and therefore $\Delta_i$ increases by $\frac{n}{4}$ for the first $\frac{n}{2}$ strings and decreases by $\frac{n}{4}$ for the last $\frac{n}{2}$ strings.

From this it follows that after $O(N / n)$ iterations, we have that $D := \max_i \Delta_i - \min_i \Delta_i \leq n$. Indeed, focus on any iteration and suppose that $D > n$; we show that $D$ decreases. We say that an index $i$ is $\alpha$-small if $\Delta_i \leq \min_j \Delta_j + \alpha \cdot \frac{n}{4}$ (i.e., if it has distance $\alpha \cdot \frac{n}{4}$ to the minimum) and we say that $i$ is $\alpha$-large if $\Delta_i \geq \max_j \Delta_j - \alpha \cdot \frac{n}{4}$ (i.e., if it has distance $\alpha \cdot \frac{n}{4}$ to the maximum). Since~\makebox{$D > n$}, the sets of $2$-small and $2$-large indices are disjoint. Therefore, either the number of $2$-small or $2$-large indices is at most $\frac{n}{2}$; we focus on the latter event, the former is symmetric. In this case clearly $\max_j \Delta_j$ decreases by $\frac{n}{4}$ in this iteration. If the minimum does not decrease, then we are done. So suppose otherwise that the minimum decreases. But this can only happen if the number of $1$-small indices exceeds $\frac{n}{2}$. In this case we analyze how many indices are 1-small in the \emph{next} iteration. As the minimum decreases, and as $\frac{n}{2}$ of the formerly $1$-small indices \emph{increase} by 1, the number of new $1$-small indices is less than $\frac{n}{2}$. In summary, after one iteration we have that both the number of $1$-small and of $1$-large indices is less than $\frac{n}{2}$, and therefore the second iteration $D$ decreases by~$\frac{n}{2}$. In both cases, on average each iteration decreases $D$ by $\frac{n}{4}$, and therefore the total number of iterations until $D \leq n$ is bounded by $O(N / n)$ (given that initially $D \leq O(N)$).

The previous paragraph proves that after the algorithm terminates, we have that $\max_i \Delta_i - \min_i \Delta_i \leq n$. In particular, the absolute difference between any $\Delta_i$ and the average $\frac{1}{n} \sum_i \Delta_i$ is at most $n$. However, note that $\sum_i \Delta_i = \sum_i (k_i + d_i) - n \bar d = 0$, and thus in fact $|\Delta_i| \leq n$ for all $i$. We have finally reached the situation from the lemma statement for $d := \bar d$.

Concerning the running time, note that each iteration can be implemented in time $O(n \log n)$. Thus, the total time is $O(N \log n)$ as claimed.
\end{proof}

\begin{lemma}[Fine-Grained Balancing] \label{lemma:balancing-fine}
Let $n \geq 3$ and let $k_1, \dots, k_n \in [N]$. Given \(k_1,\dots,k_n\), in time $O(n (N + n \log n))$ we can construct permutations $\pi_1, \dots, \pi_n, \tau \in \mathcal S_{O(N + n \log n)}$ and an integer $d$ with the following two properties:
\begin{itemize}
    \item Writing $d_i = \sum_{j \neq i} d_U(\pi_i, \pi_j)$, it holds that $|(k_i + d_i) - d| \leq 1$ for all $i \in [n]$.
    \item It holds that $d_U(\pi_1, \tau) = \dots = d_U(\pi_n, \tau)$.
\end{itemize}
\end{lemma}
\begin{proof}
Again, we will not construct the permutation directly, but instead will construct strings $a_1, \dots, a_n, b$ satisfying the following analogous properties: Writing $d_i = \sum_{j \neq i} d_H(a_i, a_j)$, we guarantee that $|d - (d_i + k_i)| \leq 1$ for some integer $d$, and $d_H(a_1, b) = \dots = d_H(a_n, b)$. Additionally, we will make sure that each string $a_1, \dots, a_n, b \in \Sigma^{O(N)}$ contains each symbol from $\Sigma$ at most once, where $\Sigma$ is some alphabet of size $O(n + N)$. In this situation we can apply the embedding from Lemma~\ref{lemma:embedding-hd-no-repeat} to obtain the permutations promised in the lemma statement.

We will first focus on the construction of the strings $a_1, \dots, a_n$. Initially, let all these strings be empty and let $\Sigma$ be an alphabet of size at least $6N + n + 1$. Then consider the following process, running for $3N$ steps: Let $I = \{i : \min_j (d_j + k_j) < d_i + k_i \}$. If $I$ has even size, then arbitrarily pair up the elements in $I$. If $I$ has odd size, we leave out one element $i \in I$ that maximizes $d_i + k_i$ from the pairing. Then let $\sigma_1, \dots, \sigma_n \in \Sigma$ be characters selected subject to the following constraints: (1) Two characters $\sigma_{i_1}$ and $\sigma_{i_2}$ are equal if and only if $(i_1, i_2)$ is a pair in the previously constructed pairing. (2) $\sigma_i$ does not appear already in $a_i$. These characters can e.g.\ be constructed as follows: Traverse the pairs $(i_1, i_2)$ in the pairing, and for each pick an arbitrary character that does not appear in the strings $a_{i_1}$ and $a_{i_2}$ and that has not been selected before in the current iteration. As~\makebox{$|a_{i_1}| = |a_{i_2}| \leq 3N$} and as $|\Sigma| > 6N + n$, there is always at least one such character available. Afterwards, assign the character $\sigma_i$ for the unpaired indices by similarly avoiding collisions. 

There is one special case: If $|I| = 1$, say $I = \{i\}$, then we instead consider the pairing consisting of exactly one pair $(i, j)$ where $j \neq i$ is arbitrary.

To analyze this process, let us write $\ell_i := k_i + d_i$. Let $\ell_{\max} := \max_i \ell_i$ and $\ell_{\min} := \min_i \ell_i$. We prove that within at most three iterations, $\ell_{\max} - \ell_{\min}$ decreases by~$1$. Indeed, focus on a single iteration. For each pair $(i_1, i_2)$ in the pairing, the quantities $\ell_{i_1}$ and $\ell_{i_2}$ increase by exactly~\makebox{$n-2$}. For each unpaired index $i$, the quantity $\ell_i$ instead increases by $n-1$. By the choice of $I$, we thus normally increase $\ell_{\min}$ by $n-1$, and thus $\ell_i - \ell_{\min}$ decreases by $1$ for each index $i$ that appears in a pair. Hence, if there is an even number of indices $i$ with $\ell_i > \ell_{\min}$, then after a single iteration we decrease $\ell_{\max} - \ell_{\min}$ by $1$. Otherwise, it can happen that one of the indices $i$ witnessing $\ell_{\max}$ is left out of the pairing. In this case, however, $\ell_i - \ell_{\min}$ (and thereby $\ell_{\max} - \ell_{\min}$) decreases by~$1$ in the next iteration. The only exception is when $|I| = 1$. In this case the algorithm is forced to increase $\ell_{\min}$ by only $n-2$. But it is not hard to verify that also in this case, using that~\makebox{$n \geq 3$} and assuming that $\ell_{\max} - \ell_{\min} \geq 2$, after at most 2 more iterations we have decreased our progress measure $\ell_{\max} - \ell_{\min}$ by at least~$1$. Recall that initially $0 \leq \ell_{\min} \leq \ell_{\max} \leq N$, and thus the process correctly terminates after~$3N$ iterations. At this point we assign $d := \ell_{\min}$.

Finally, we comment on the construction of $b$: Simply take $b$ to any permutation of $3N$ symbols disjoint from the symbols used to construct the strings $a_1, \dots, a_n$; then clearly $d_H(a_1, b) = \dots = d_H(a_n, b)$.

The process from before can be implemented in time $O(n N)$. The application of Lemma~\ref{lemma:embedding-hd-no-repeat} to embed from the Hamming metric into the Ulam metric afterwards take time $O(n (|\Sigma| + n \log n)) = O(n (N + n \log n))$.
\end{proof}

\begin{cor}[Full Balancing] \label{cor:balancing}
Let $n$ and \(L\) be integers such that \(n\) is divisible by $4$ and let $k_1, \dots, k_n \in [O(nL)]$. Given \(k_1,\dots,k_n\), in time $\Tilde{O}(n^2+nL)$ we can construct permutations $\pi_1, \dots, \pi_n, \tau \in \mathcal S_{\Tilde{O}(n+L)}$ and an integer $d$ with the following two properties:
\begin{itemize}
    \item Writing $d_i = \sum_{j \neq i} d_U(\pi_i, \pi_j)$, it holds that $|(k_i + d_i) - d| \leq 1$ for all $i \in [n]$.
    \item It holds that $d_U(\pi_1, \tau) = \dots = d_U(\pi_n, \tau)$.
\end{itemize}
\end{cor}
\begin{proof}
First apply the coarse-grained balancing (Lemma~\ref{lemma:balancing-coarse}) on the given values $k_1, \dots, k_n$ and for $N = O(nL)$ to obtain permutations $\pi_1', \dots, \pi_n', \tau' \in \mathcal S_{ O(L)}$. Let $k_i' := k_i + \sum_{j \neq i} d_U(\pi_i', \pi_j')$; then Lemma~\ref{lemma:balancing-coarse} guarantees that all $k_i'$ equals some common integer $d'$ up to an additive error of~$n$. Next, apply the fine-grained balancing (Lemma~\ref{lemma:balancing-fine}) on the values $k_i'$ and for $N = n$ to obtain permutations $\pi_1'', \dots, \pi_n'', \tau'' \in \mathcal S_{\tilde O(n)}$. Lemma~\ref{lemma:balancing-fine} guarantees that all $k_i'' := k_i' + \sum_{j \neq i} d_U(\pi_i'', \pi_j'')$ equal some common integer $d''$, up to an additive error of $1$. The claimed statement follows by choosing the concatenations $\pi_i := \pi_i' \pi_i''$ and $\tau := \tau' \tau''$ (while using fresh symbols as is necessary to ensure that the strings are permutations) and picking $d := d''$.
\end{proof}

\subsection{Hardness for Monochromatic Instances} \label{sec:median:sec:monochromatic}

\begin{theorem} \label{thm:monochromatic-ulam-median}
    Let $\varepsilon > 0$ and $\alpha \geq 1$. There is no algorithm running in time $O((n^2L)^{1-\varepsilon})$ that solves the (\textsc{Monochromatic}) \dm problem for $n$ permutations of length $L = \Theta(n^{\alpha})$, unless the SETH fails.
\end{theorem}
\begin{proof}
We follow the same reduction as in Theorem~\ref{thm:bichromatic-ulam-median}, which, given an initial $4$-OV instance produces an instance $(X, Y)$ of \textsc{Bichromatic} \dm on $n$ permutations of length $L$. During the reduction, we can set everything up so that \(L=\Omega(n)\). Without loss of generality, we may also assume that $n$ is divisible by~$4$ by adding some dummy vectors in the initial $4$-OV instance if necessary. Let $x_1, \dots, x_n$ denote the permutations in $X$ and let $y_1, \dots, y_n$ denote the permutations in $Y$. As a first preprocessing step, concatenate each permutation $x_i$ (or~$y_i$) two times to itself (using fresh symbols when necessary) so that each permutation becomes of length \(3L\) and the median distance $\min_i \sum_j d_U(x_i, y_j)$ becomes a multiple of $3$.

We first make the observation that the Ulam distance between any two permutations in \(X\) can be computed very quickly --- in time that is proportional to the square root of the length of the permutations.

\begin{obs}
\label{obs:fast}
    Let \(x, x'\in X\). Then \(d_U(x, x')\) can be computed in time \(O(L^{1/2+o(1)})\).
\end{obs}
\begin{proof}
    \renewcommand{\qedsymbol}{\(\lrcorner\)}
    Since \(x\in X\), there exist \(T=\{t_1, t_2, \ldots , t_{m}\}\subseteq \{0, 1\}^d\) such that \(x=f(T)\circ f(T)\circ f(T)\), where \(f\) is the function from Theorem~\ref{thm:ov-to-ulam} and each concatenation is done using a fresh set of symbols. Similarly, there exist \(T'=\{t'_1, t'_2, \ldots , t'_{m}\}\subseteq \{0, 1\}^d\) such that \(x'=f(T')\circ f(T')\circ f(T')\), where again each concatenation is with a fresh set of symbols.

    Therefore, we have~\smash{\(d_U(x, x')=3{m}\sum_{i\in [m]}d_H(t_i, t'_i)\)}. Thus, \(d_U(x, x')\) can be found by computing the Hamming distance between two bit strings of length \(md\). Further recall that \(L=O(m^2d)\). Therefore, \(md= O(L^{1/2+o(1)})\) and the conclusion follows.
\end{proof}

By Observation~\ref{obs:fast}, we can compute in time $O(n^2 L^{1/2+o(1)})$ all pairwise distances $d_U(x_i, x_j)$. Let $k_i := \sum_{j \neq i} d_U(x_i, x_j)$; clearly we have that $k_i \leq 3 n L$. Thus, we may apply Corollary~\ref{cor:balancing} to obtain permutations $\pi_1, \dots, \pi_n, \tau \in \mathcal S_{L'}$, where~\smash{$L' = \tilde O(n+L)=\Tilde{O}(L)$}, with
\begin{equation*}
    |(k_i + \sum_{j \neq i} d_U(\pi_i, \pi_j)) - D| \leq 1
\end{equation*}
for some integer $D$ and for all $i \in [n]$, and with $M := d_U(\pi_1, \tau) = \dots = d_U(\pi_n, \tau)$.

Additionally, let $K := 10 (3L + L')$. Compute some length-$O(K)$ permutations $\mu, \eta_1, \dots, \eta_n$ such that $d_U(\mu, \eta_i) = K$ and such that $\sum_j d_U(\eta_i, \eta_j) = nK$. For instance, viewing $\mu, \eta_1, \dots, \eta_n$ as $0$-$1$-strings under the Hamming distance to be embedded by Lemma~\ref{lemma:embedding-hd-small-alphabet}, take $\mu$ to be the all-$0$ string of length $2K$, let half of the strings $\eta_i$ be the string $0^K 1^K$ and let the other half of the strings~$\eta_i$ be the string $1^K 0^K$.

We are now ready to construct the \textsc{Monochromatic Discrete Ulam Median} instance~$Z$. We include into $Z$ the following $2n$ permutations:\footnote{As always, we use fresh symbols when necessary to ensure that the resulting strings are permutations.}
\begin{itemize}
    \item $x_i' := x_i \, \pi_i \, \mu$ (for $i \in [n]$), and
    \item $y_i' := y_i \, \tau \, \eta_i$ (for $i \in [n]$).
\end{itemize}

We claim that this construction is correct in the following sense: (1) All  discrete medians of~$Z$ are strings $x_i'$. (2) Whenever $x_i'$ is a discrete median in $Z$, then $x_i$ is a discrete median in $(X, Y)$. (3)~There is some discrete median $x_i$ in $(X, Y)$ such that $x_i'$ is a discrete median in $Z$. The proofs of all three claims easily follow from the following calculations. On the one hand, the median distance for each $x_i'$ is
\begin{align*}
    \sum_{z \in Z} d_U(x_i', z)
    &= \sum_j d_U(x_i', x_j') + \sum_j d_U(x_i', y_j') \\
    &= \sum_j (d_U(x_i, x_j) + d_U(\pi_i, \pi_j) + d_U(\mu, \mu)) + \sum_j (d_U(x_i, y_j) + d_U(\pi_i, \tau) + d_U(\mu, \eta_i)) \\
    &= k_i + \sum_j d_U(\pi_i, \pi_j) + \sum_j d_U(x_i, y_j) + n M + n K \\
    &= \sum_j d_U(x_i, y_j) + n M + n K + D \pm 1.
\intertext{On the other hand, the median distance for each $y_i'$ is}
    \sum_{z \in Z} d_U(y_i', z)
    &= \sum_j d_U(y_i', x_j') + \sum_j d_U(y_i', y_j') \\
    &\geq \sum_j d_U(\eta_i, \mu) + \sum_j d_U(\eta_i, \eta_j) \\
    &= 2nK.
\end{align*}
Comparing these two terms, and recalling that $\sum_j d_U(x_i, y_j) + n M + D \leq 3 n L + n L' + n L' < n K$, it is clear that the median distance of any $y_i'$ is always significantly larger that the median distance of any $x_i'$ proving (1). Recalling further that all median distances $\sum_j d_U(x_i, y_j)$ in the original instance are multiples of $3$, the $\pm 1$ term in the first computation becomes irrelevant, completing the proofs of claims (2) and (3).

Finally we comment on the running time. The original reduction, along with the computation of the values $k_i$ takes time $\tilde O(nL+n^2 L^{1/2+o(1)})$. Running Corollary~\ref{cor:balancing} takes time $\tilde O(n^2+nL)$, and the final instance $Z$ can be implemented in negligible overhead.
\end{proof}

\section*{Acknowledgements}
The authors would like to thank the \emph{DIMACS Workshop on Efficient Algorithms for High Dimensional Metrics: New Tools} for a special collaboration opportunity. We would also like to thank Siam Habib and Antti Roeyskoe for helpful discussions.
Part of this work was done while Mursalin Habib and Karthik C.\ S.\ were visiting INSAIT, Sofia University ``St.\ Kliment Ohridski''  and were partially funded by the Ministry of Education and Science of Bulgaria’s support for
INSAIT as part of the Bulgarian National Roadmap for Research Infrastructure. They were also supported by the National Science Foundation under Grants CCF-2313372,  CCF-2422558, and CCF-2443697. 
\bibliographystyle{alpha}
\bibliography{refs}

\newcommand{\etalchar}[1]{$^{#1}$}
\begin{thebibliography}{RMR{\etalchar{+}}17}

\bibitem[ABC{\etalchar{+}}23]{AbboudBCSS23}
Amir Abboud, MohammadHossein Bateni, Vincent Cohen{-}Addad, {Karthik {C. S.}},
  and Saeed Seddighin.
\newblock On complexity of 1-center in various metrics.
\newblock In Nicole Megow and Adam~D. Smith, editors, {\em Approximation,
  Randomization, and Combinatorial Optimization. Algorithms and Techniques,
  {APPROX/RANDOM} 2023, September 11-13, 2023, Atlanta, Georgia, {USA}}, volume
  275 of {\em LIPIcs}, pages 1:1--1:19. Schloss Dagstuhl - Leibniz-Zentrum
  f{\"{u}}r Informatik, 2023.

\bibitem[ABHS22]{AbboudBHS22}
Amir Abboud, Karl Bringmann, Danny Hermelin, and Dvir Shabtay.
\newblock Scheduling lower bounds via {AND} subset sum.
\newblock {\em J. Comput. Syst. Sci.}, 127:29--40, 2022.

\bibitem[ABV15]{AbboudBW15}
Amir Abboud, Arturs Backurs, and Virginia {Vassilevska Williams}.
\newblock Tight hardness results for {LCS} and other sequence similarity
  measures.
\newblock In Venkatesan Guruswami, editor, {\em {IEEE} 56th Annual Symposium on
  Foundations of Computer Science, {FOCS} 2015, Berkeley, CA, USA, 17-20
  October, 2015}, pages 59--78. {IEEE} Computer Society, 2015.

\bibitem[ACN08]{AilonCN08}
Nir Ailon, Moses Charikar, and Alantha Newman.
\newblock Aggregating inconsistent information: Ranking and clustering.
\newblock {\em J. {ACM}}, 55(5):23:1--23:27, 2008.

\bibitem[AD99]{AldousD99}
David Aldous and Persi Diaconis.
\newblock Longest increasing subsequences: from patience sorting to the
  {Baik}-{Deift}-{Johansson} theorem.
\newblock {\em Bulletin of the American Mathematical Society}, 36(4):413--432,
  1999.

\bibitem[ADV{\etalchar{+}}25]{AlmanDWXXZ25}
Josh Alman, Ran Duan, Virginia {Vassilevska Williams}, Yinzhan Xu, Zixuan Xu,
  and Renfei Zhou.
\newblock More asymmetry yields faster matrix multiplication.
\newblock In Yossi Azar and Debmalya Panigrahi, editors, {\em Proceedings of
  the 2025 Annual {ACM-SIAM} Symposium on Discrete Algorithms, {SODA} 2025, New
  Orleans, LA, USA, January 12-15, 2025}, pages 2005--2039. {SIAM}, 2025.

\bibitem[AFG{\etalchar{+}}23]{AbboudFGSS23}
Amir Abboud, Nick Fischer, Elazar Goldenberg, {Karthik {C. S.}}, and Ron
  Safier.
\newblock Can you solve closest string faster than exhaustive search?
\newblock In Inge~Li G{\o}rtz, Martin Farach{-}Colton, Simon~J. Puglisi, and
  Grzegorz Herman, editors, {\em 31st Annual European Symposium on Algorithms,
  {ESA} 2023, September 4-6, 2023, Amsterdam, The Netherlands}, volume 274 of
  {\em LIPIcs}, pages 3:1--3:17. Schloss Dagstuhl - Leibniz-Zentrum f{\"{u}}r
  Informatik, 2023.

\bibitem[BBD09]{BiedlBD09}
Therese Biedl, Franz{-}Josef Brandenburg, and Xiaotie Deng.
\newblock On the complexity of crossings in permutations.
\newblock {\em Discret. Math.}, 309(7):1813--1823, 2009.

\bibitem[BBGH15]{BachmaierBGH15}
Christian Bachmaier, Franz~J. Brandenburg, Andreas Glei{\ss}ner, and Andreas
  Hofmeier.
\newblock On the hardness of maximum rank aggregation problems.
\newblock {\em J. Discrete Algorithms}, 31:2--13, 2015.

\bibitem[BC20]{BringmannC20}
Karl Bringmann and Bhaskar~Ray Chaudhury.
\newblock Polyline simplification has cubic complexity.
\newblock {\em J. Comput. Geom.}, 11(2):94--130, 2020.

\bibitem[BCE{\etalchar{+}}16]{BrandtCELP16}
Felix Brandt, Vincent Conitzer, Ulle Endriss, J{\'{e}}r{\^{o}}me Lang, and
  Ariel~D. Procaccia, editors.
\newblock {\em Handbook of Computational Social Choice}.
\newblock Cambridge University Press, 2016.

\bibitem[BI18]{BackursI18}
Arturs Backurs and Piotr Indyk.
\newblock Edit distance cannot be computed in strongly subquadratic time
  (unless {SETH} is false).
\newblock {\em {SIAM} J. Comput.}, 47(3):1087--1097, 2018.

\bibitem[BK15]{BringmannK15}
Karl Bringmann and Marvin K{\"{u}}nnemann.
\newblock Quadratic conditional lower bounds for string problems and dynamic
  time warping.
\newblock In Venkatesan Guruswami, editor, {\em {IEEE} 56th Annual Symposium on
  Foundations of Computer Science, {FOCS} 2015, Berkeley, CA, USA, 17-20
  October, 2015}, pages 79--97. {IEEE} Computer Society, 2015.

\bibitem[BK19]{bringmann2019lectures}
Karl Bringmann and Marvin K\"unnemann.
\newblock Quadratic lower bounds for sequence similarity, 2019.
\newblock Available at
  \url{https://www.mpi-inf.mpg.de/fileadmin/inf/d1/teaching/summer19/finegrained/lec3.pdf}.

\bibitem[BKN21]{BKN21}
Boris Bukh, {Karthik {C. S.}}, and Bhargav Narayanan.
\newblock Applications of random algebraic constructions to hardness of
  approximation.
\newblock In {\em 62nd {IEEE} Annual Symposium on Foundations of Computer
  Science, {FOCS} 2021, Denver, CO, USA, February 7-10, 2022}, pages 237--244.
  {IEEE}, 2021.

\bibitem[CDK21]{ChakrabortyDK21}
Diptarka Chakraborty, Debarati Das, and Robert Krauthgamer.
\newblock Approximating the median under the ulam metric.
\newblock In D{\'{a}}niel Marx, editor, {\em Proceedings of the 2021 {ACM-SIAM}
  Symposium on Discrete Algorithms, {SODA} 2021, Virtual Conference, January 10
  - 13, 2021}, pages 761--775. {SIAM}, 2021.

\bibitem[CDK23]{Chakraborty0K23}
Diptarka Chakraborty, Debarati Das, and Robert Krauthgamer.
\newblock Clustering permutations: New techniques with streaming applications.
\newblock In Yael~Tauman Kalai, editor, {\em 14th Innovations in Theoretical
  Computer Science Conference, {ITCS} 2023, January 10-13, 2023, MIT,
  Cambridge, Massachusetts, {USA}}, volume 251 of {\em LIPIcs}, pages
  31:1--31:24. Schloss Dagstuhl - Leibniz-Zentrum f{\"{u}}r Informatik, 2023.

\bibitem[CDKS22]{ChakrabortyD0S22}
Diptarka Chakraborty, Syamantak Das, Arindam Khan, and Aditya Subramanian.
\newblock Fair rank aggregation.
\newblock In Sanmi Koyejo, S.~Mohamed, A.~Agarwal, Danielle Belgrave, K.~Cho,
  and A.~Oh, editors, {\em Advances in Neural Information Processing Systems
  35: Annual Conference on Neural Information Processing Systems 2022, NeurIPS
  2022, New Orleans, LA, USA, November 28 - December 9, 2022}, 2022.

\bibitem[CGI{\etalchar{+}}16]{CarmosinoGIMPS16}
Marco~L. Carmosino, Jiawei Gao, Russell Impagliazzo, Ivan Mihajlin, Ramamohan
  Paturi, and Stefan Schneider.
\newblock Nondeterministic extensions of the strong exponential time hypothesis
  and consequences for non-reducibility.
\newblock In Madhu Sudan, editor, {\em Proceedings of the 2016 {ACM} Conference
  on Innovations in Theoretical Computer Science, Cambridge, MA, USA, January
  14-16, 2016}, pages 261--270. {ACM}, 2016.

\bibitem[CGJ21]{ChakrabortyGJ21}
Diptarka Chakraborty, Kshitij Gajjar, and Agastya~Vibhuti Jha.
\newblock Approximating the center ranking under ulam.
\newblock In Mikolaj Bojanczyk and Chandra Chekuri, editors, {\em 41st {IARCS}
  Annual Conference on Foundations of Software Technology and Theoretical
  Computer Science, {FSTTCS} 2021, December 15-17, 2021, Virtual Conference},
  volume 213 of {\em LIPIcs}, pages 12:1--12:21. Schloss Dagstuhl -
  Leibniz-Zentrum f{\"{u}}r Informatik, 2021.

\bibitem[CMS01]{CormodeMS01}
Graham Cormode, S.~Muthukrishnan, and S{\"{u}}leyman~Cenk Sahinalp.
\newblock Permutation editing and matching via embeddings.
\newblock In Fernando Orejas, Paul~G. Spirakis, and Jan van Leeuwen, editors,
  {\em Automata, Languages and Programming, 28th International Colloquium,
  {ICALP} 2001, Crete, Greece, July 8-12, 2001, Proceedings}, volume 2076 of
  {\em Lecture Notes in Computer Science}, pages 481--492. Springer, 2001.

\bibitem[DKL19]{DKL19}
Roee David, {Karthik {C. S.}}, and Bundit Laekhanukit.
\newblock On the complexity of closest pair via polar-pair of point-sets.
\newblock {\em {SIAM} J. Discret. Math.}, 33(1):509--527, 2019.

\bibitem[DKNS01]{DworkKNS01}
Cynthia Dwork, Ravi Kumar, Moni Naor, and D.~Sivakumar.
\newblock Rank aggregation methods for the web.
\newblock In Vincent~Y. Shen, Nobuo Saito, Michael~R. Lyu, and Mary~Ellen
  Zurko, editors, {\em Proceedings of the Tenth International World Wide Web
  Conference, {WWW} 10, Hong Kong, China, May 1-5, 2001}, pages 613--622.
  {ACM}, 2001.

\bibitem[dlHC00]{HigueraC00}
Colin de~la Higuera and Francisco Casacuberta.
\newblock Topology of strings: Median string is np-complete.
\newblock {\em Theor. Comput. Sci.}, 230(1-2):39--48, 2000.

\bibitem[FKS03]{FaginKS03}
Ronald Fagin, Ravi Kumar, and D.~Sivakumar.
\newblock Efficient similarity search and classification via rank aggregation.
\newblock In Alon~Y. Halevy, Zachary~G. Ives, and AnHai Doan, editors, {\em
  Proceedings of the 2003 {ACM} {SIGMOD} International Conference on Management
  of Data, San Diego, California, USA, June 9-12, 2003}, pages 301--312. {ACM},
  2003.

\bibitem[FL97]{FrancesL97}
Moti Frances and Ami Litman.
\newblock On covering problems of codes.
\newblock {\em Theory Comput. Syst.}, 30(2):113--119, 1997.

\bibitem[GBC{\etalchar{+}}13]{GoldmanBCDLSB13}
Nick Goldman, Paul Bertone, Siyuan Chen, Christophe Dessimoz, Emily~M.
  LeProust, Botond Sipos, and Ewan Birney.
\newblock Towards practical, high-capacity, low-maintenance information storage
  in synthesized {DNA}.
\newblock {\em Nat.}, 494(7435):77--80, 2013.

\bibitem[Gus97]{Gusfield1997}
Dan Gusfield.
\newblock {\em Algorithms on Strings, Trees, and Sequences - Computer Science
  and Computational Biology}.
\newblock Cambridge University Press, 1997.

\bibitem[Har92]{Harman92a}
Donna Harman.
\newblock Ranking algorithms.
\newblock In William~B. Frakes and Ricardo~A. Baeza{-}Yates, editors, {\em
  Information Retrieval: Data Structures {\&} Algorithms}, pages 363--392.
  Prentice-Hall, 1992.

\bibitem[IPZ01]{impagliazzo2001problems}
Russell Impagliazzo, Ramamohan Paturi, and Francis Zane.
\newblock Which problems have strongly exponential complexity?
\newblock {\em Journal of Computer and System Sciences}, 63(4):512--530, 2001.

\bibitem[Kar72]{Karp72}
Richard~M. Karp.
\newblock Reducibility among combinatorial problems.
\newblock In Raymond~E. Miller and James~W. Thatcher, editors, {\em Proceedings
  of a symposium on the Complexity of Computer Computations, held March 20-22,
  1972, at the {IBM} Thomas J. Watson Research Center, Yorktown Heights, New
  York, {USA}}, The {IBM} Research Symposia Series, pages 85--103. Plenum
  Press, New York, 1972.

\bibitem[Kem59]{Kemeny59}
John~G. Kemeny.
\newblock Mathematics without numbers.
\newblock {\em Daedalus}, 88(4):577--591, 1959.

\bibitem[KM20]{KM20}
{Karthik {C. S.}} and Pasin Manurangsi.
\newblock On closest pair in euclidean metric: Monochromatic is as hard as
  bichromatic.
\newblock {\em Comb.}, 40(4):539--573, 2020.

\bibitem[Koh85]{Kohonen85}
Teuvo Kohonen.
\newblock Median strings.
\newblock {\em Pattern Recognit. Lett.}, 3(5):309--313, 1985.

\bibitem[KS07]{Kenyon-MathieuS07}
Claire Kenyon{-}Mathieu and Warren Schudy.
\newblock How to rank with few errors.
\newblock In David~S. Johnson and Uriel Feige, editors, {\em Proceedings of the
  39th Annual {ACM} Symposium on Theory of Computing, San Diego, California,
  USA, June 11-13, 2007}, pages 95--103. {ACM}, 2007.

\bibitem[LLM{\etalchar{+}}03]{LanctotLMWZ03}
J.~Kevin Lanct{\^{o}}t, Ming Li, Bin Ma, Shaojiu Wang, and Louxin Zhang.
\newblock Distinguishing string selection problems.
\newblock {\em Inf. Comput.}, 185(1):41--55, 2003.

\bibitem[LLQ{\etalchar{+}}07]{LiuLQML07}
Yuting Liu, Tie{-}Yan Liu, Tao Qin, Zhiming Ma, and Hang Li.
\newblock Supervised rank aggregation.
\newblock In Carey~L. Williamson, Mary~Ellen Zurko, Peter~F. Patel{-}Schneider,
  and Prashant~J. Shenoy, editors, {\em Proceedings of the 16th International
  Conference on World Wide Web, {WWW} 2007, Banff, Alberta, Canada, May 8-12,
  2007}, pages 481--490. {ACM}, 2007.

\bibitem[LWX19]{Li19}
Xue Li, Xinlei Wang, and Guanghua Xiao.
\newblock A comparative study of rank aggregation methods for partial and top
  ranked lists in genomic applications.
\newblock {\em Briefings in bioinformatics}, 20(1):178--189, 2019.

\bibitem[MJC00]{Martinez-HinarejosJC00}
Carlos~D. Mart{\'{\i}}nez{-}Hinarejos, Alfons Juan, and Francisco Casacuberta.
\newblock Use of median string for classification.
\newblock In {\em 15th International Conference on Pattern Recognition,
  ICPR'00, Barcelona, Spain, September 3-8, 2000}, pages 2903--2906. {IEEE}
  Computer Society, 2000.

\bibitem[NR03]{NicolasR03}
Fran{\c{c}}ois Nicolas and Eric Rivals.
\newblock Complexities of the centre and median string problems.
\newblock In Ricardo~A. Baeza{-}Yates, Edgar Ch{\'{a}}vez, and Maxime
  Crochemore, editors, {\em Combinatorial Pattern Matching, 14th Annual
  Symposium, {CPM} 2003, Morelia, Michoc{\'{a}}n, Mexico, June 25-27, 2003,
  Proceedings}, volume 2676 of {\em Lecture Notes in Computer Science}, pages
  315--327. Springer, 2003.

\bibitem[ODL{\etalchar{+}}20]{OliveiraDLMP20}
Samuel E.~L. Oliveira, Victor Diniz, An{\'{\i}}sio Lacerda, Luiz H.~C.
  Merschmann, and Gisele~L. Pappa.
\newblock Is rank aggregation effective in recommender systems? an experimental
  analysis.
\newblock {\em {ACM} Trans. Intell. Syst. Technol.}, 11(2):16:1--16:26, 2020.

\bibitem[Pev00]{Pevzner00}
Pavel~A. Pevzner.
\newblock {\em Computational molecular biology - an algorithmic approach}.
\newblock {MIT} Press, 2000.

\bibitem[Pop07]{Popov07}
V.~Y. Popov.
\newblock Multiple genome rearrangement by swaps and by element duplications.
\newblock {\em Theor. Comput. Sci.}, 385(1-3):115--126, 2007.

\bibitem[RMR{\etalchar{+}}17]{RashtchianMRAJY17}
Cyrus Rashtchian, Konstantin Makarychev, Mikl{\'{o}}s~Z. R{\'{a}}cz, Siena Ang,
  Djordje Jevdjic, Sergey Yekhanin, Luis Ceze, and Karin Strauss.
\newblock Clustering billions of reads for {DNA} data storage.
\newblock In Isabelle Guyon, Ulrike von Luxburg, Samy Bengio, Hanna~M. Wallach,
  Rob Fergus, S.~V.~N. Vishwanathan, and Roman Garnett, editors, {\em Advances
  in Neural Information Processing Systems 30: Annual Conference on Neural
  Information Processing Systems 2017, December 4-9, 2017, Long Beach, CA,
  {USA}}, pages 3360--3371, 2017.

\bibitem[Sch61]{schensted1961longest}
Craige Schensted.
\newblock Longest increasing and decreasing subsequences.
\newblock {\em Canadian Journal of mathematics}, 13:179--191, 1961.

\bibitem[vZW07]{ZuylenW07}
Anke van Zuylen and David~P. Williamson.
\newblock Deterministic algorithms for rank aggregation and other ranking and
  clustering problems.
\newblock In Christos Kaklamanis and Martin Skutella, editors, {\em
  Approximation and Online Algorithms, 5th International Workshop, {WAOA} 2007,
  Eilat, Israel, October 11-12, 2007. Revised Papers}, volume 4927 of {\em
  Lecture Notes in Computer Science}, pages 260--273. Springer, 2007.

\bibitem[Wil05]{Williams05}
Ryan Williams.
\newblock A new algorithm for optimal 2-constraint satisfaction and its
  implications.
\newblock {\em Theor. Comput. Sci.}, 348(2-3):357--365, 2005.

\bibitem[WSCK12]{WangSCK12}
Tiance Wang, John Sturm, Paul~W. Cuff, and Sanjeev~R. Kulkarni.
\newblock Condorcet voting methods avoid the paradoxes of voting theory.
\newblock In {\em 50th Annual Allerton Conference on Communication, Control,
  and Computing, Allerton 2012, Allerton Park {\&} Retreat Center, Monticello,
  IL, USA, October 1-5, 2012}, pages 201--203. {IEEE}, 2012.

\bibitem[YL78]{YoungL78}
H.~Peyton Young and Arthur Levenglick.
\newblock A consistent extension of condorcet’s election principle.
\newblock {\em SIAM Journal on applied Mathematics}, 35(2):285--300, 1978.

\bibitem[You88]{Young88}
H.~Peyton Young.
\newblock Condorcet's theory of voting.
\newblock {\em American Political science review}, 82(4):1231--1244, 1988.

\end{thebibliography}

\appendix

\section{Proof of Lemma~\ref{lemma:suboptimality}}
\label{sec:suboptimality}

In this section, we prove Lemma~\ref{lemma:suboptimality}. Before we do so, we first prove the following helper lemma.

    \begin{lemma}
        \label{lemma:same-last}
        Let \(\pi, \pi'\in \mathcal{S}_n\) be permutations with the same last symbol, i.e., \(\pi[n]=\pi'[n]=x\) for some \(x\in [n]\). Then for any median \(\pi^*\) of \(\pi\) and \(\pi'\) (i.e., any \(\pi^*\in \mathcal{S}_n\) minimizing the expression \(d_U(\pi^*, \pi)+d_U(\pi^*, \pi')\)), we must have \(\pi^*[n]=x\).
    \end{lemma}
    \begin{proof}
        Let \(\pi^*\) be a median of \(\pi\) and \(\pi'\) and assume for contradiction that \(x^*=\pi^*[n] \neq x\). Let~\(\rho\) be any longest common subsequence of \(\pi\) and \(\pi^*\). Similarly, let \(\rho'\) be any longest common subsequence of \(\pi'\) and \(\pi^*\). Note that if \(x\) does not appear in at least one of \(\rho\) and \(\rho'\), then \(\pi^*\) cannot be a median. Indeed, assume without loss of generality that \(x\) does not appear in \(\rho\). Let \(\Tilde{\pi}\) be the permutation obtained from \(\pi^*\) by moving \(x\) to the last position. Then clearly, \(\rho \circ x\) is a common subsequence of \(\pi\) and \(\Tilde{\pi}\) and therefore, \(d_U(\Tilde{\pi}, \pi) <d_U(\pi^*, \pi)\). Furthermore, we claim that \(\rho'\) remains a common subsequence of \(\pi'\) and \(\Tilde{\pi}\). There are two cases to consider. In the case when \(x\) does not appear in \(\rho'\), the claim is immediate since moving \(x\) does not affect the relative order of the symbols in \(\rho'\). Additionally, if \(\rho'\) does contain \(x\), then it must be the last symbol of \(\rho'\). Thus, even in this case, moving \(x\) to the end does not affect the relative order of the symbols in \(\rho'\) and \(\rho'\) remains a common subsequence. Combining these two observations, we have \(d_U(\Tilde{\pi}, \pi)+d_U(\Tilde{\pi}, \pi')<d_U(\pi^*, \pi)+d_U(\pi^*, \pi')\), contradicting the assumption that \(\pi^*\) is a median of \(\pi\) and \(\pi'\).

        If \(x\) appears in both \(\rho\) and \(\rho'\), then \(x^*\) appears in neither \(\rho\) nor \(\rho'\).  Once again, we will construct from \(\pi^*\) a new permutation \({\pi}^{\dagger}\) with better median cost. Denote by \(\Sigma_{\rho}\)  the set of symbols appearing in \(\rho\). Note that since \(x^*\) does not appear in \(\rho\) and \(\rho\) is a longest common subsequence of \(\pi\) and \(\pi^*\), we must have \(\pi^*|_{\Sigma_{\rho}\cup\{x^*\}} \neq \pi|_{\Sigma_{\rho}\cup\{x^*\}}\).\footnote{Recall that for any permutation \(\pi\in \mathcal{S}_n\) and any subset \(A\subseteq [n]\) of the symbols, we denote by \(\pi|_{A}\) to be the string obtained from \(\pi\) by deleting all the symbols not in \(A\).} Let \(\pi^\dagger\) be any permutation obtained from \(\pi^*\) by only moving the symbol \(x^*\) to any position such that the strings \({\pi}^\dagger|_{\Sigma_\rho\cup\{x^*\}}\) and \(\pi|_{\Sigma_\rho\cup\{x^*\}}\) become equal. Clearly, \(d_U(\pi^\dagger, \pi) < d_U(\pi^*, \pi)\) since \(\pi|_{\Sigma_\rho\cup\{x^*\}}\) is a common subsequence of \(\pi\) and \(\pi^\dagger\). Furthermore, \(\rho'\) remains a common subsequence of \(\pi'\) and \(\pi^\dagger\) since we have only moved the symbol \(x^*\) which does not appear in \(\rho'\). Thus, \(d_U(\pi^\dagger, \pi')\leq d_U(\pi^*, \pi')\), contradicting the fact that \(\pi^*\) is a median of \(\pi\) and \(\pi'\).
    \end{proof}

Now we are ready to prove Lemma~\ref{lemma:suboptimality}.

\tedious*

\begin{proof}
    Let \(\pi^*\) be any median of \(\pi^L\) and \(\pi^R\). We will actually prove that \(\pi^*\in \mathcal{S}_{N}^*\). Since by Lemma~\ref{lemma:legal-cost}, we have \(d_U(\pi, \pi^L) + d_U(\pi, \pi^R) = n\), this will complete the proof. The proof will be in two parts. Let \(\pi^{L'}, \pi^{R'} \in \mathcal{S}_{2n+1}\) be the permutations obtained by deleting the suffix \(X_2\) from \(\pi^L\) and \(\pi^R\), respectively. First we will show that \(\pi^*\) must have a median of \(\pi^{L'}\) and \(\pi^{R'}\) as its prefix, followed by \(X_2\).
    \begin{lemma}
        \label{lemma:median-struct-1}
        We must have \(\pi^*=\pi^{*'}\circ X_2\), where \(\pi^{*'}\) is a median of \(\pi^{L'}\) and \(\pi^{R'}\). 
    \end{lemma}
    \begin{proof}
                \renewcommand{\qedsymbol}{\(\lrcorner\)}

                By Lemma~\ref{lemma:same-last}, the last symbol of \(\pi^*\) is the same as those of \(\pi^L\) and \(\pi^R\), namely \((3n+2)\). Furthermore, if we delete this last symbol from \(\pi^*\), \(\pi^L\) and \(\pi^R\), the first string still remains the median of the latter two, since otherwise we would be able to find a permutation with strictly better median cost than \(\pi^*\). The conclusion follows by repeating this argument \(|X_2|=n+1\) times.
    \end{proof}
    Let \(\pi^{*'}\in \mathcal{S}_{2n+1}\) be the permutation obtained by deleting the suffix \(X_2\) from \(\pi^*\). Next, we will show that \(\pi^{*'}\) must necessarily take a very specific form -- some subset of \([n]\) in increasing order, followed by \(X_1\), followed by the rest of the symbol in decreasing order.
    \begin{lemma}
        \label{lemma:median-struct-2}
        There exists sets \(A=\{a_1, a_2, \ldots , a_r\}\), \(B=\{b_1, b_2, \ldots, b_{n-r}\}\) with \(A\sqcup B = [n]\), \(a_1<a_2<\cdots < a_r\) and \(b_1>b_2>\cdots > b_{n-r}\) such that:
        \[\pi^{*'}=a_1\circ a_2 \circ \cdots \circ a_r \circ X_1 \circ b_1 \circ b_2 \circ \cdots \circ b_{n-r}.\]
    \end{lemma}
    \begin{proof}
        \renewcommand{\qedsymbol}{\(\lrcorner\)}
        Let \(\rho^{L'}\) be any longest common subsequence of \(\pi^{L'}\) and \(\pi^{*'}\). Similarly, let \(\rho^{R'}\) be any longest common subsequence of \(\pi^{R'}\) and \(\pi^{*'}\). We claim that there does not exist a symbol \(x\in [n]\) that appears in both \(\rho^{L'}\) and \(\rho^{R'}\). Suppose otherwise. First note that there cannot exist another symbol \(y\in [n]\) that also appears in both \(\rho^{L'}\) and \(\rho^{R'}\). This is because the relative order of \(x\) and \(y\) is different in \(\pi^{L'}\) and \(\pi^{R'}\). Furthermore, every symbol \(z\) in \(X_1\) must appear in at most one of \(\rho^{L'}\) and \(\rho^{R'}\) -- if \(z\) appears to the left of \(x\) in \(\pi^*\), then it cannot appear in \(\rho^{L'}\), and if \(z\) appears to the right of \(x\) in \(\pi^{*'}\), then it cannot appear in \(\rho^{R'}\). Therefore, if there does exist \(x\in [n]\) that appears in both \(\rho^{L'}\) and \(\rho^{R'}\), then we would have:
        \begin{align*}\allowdisplaybreaks
            |\rho^{L'}|  + |\rho^{R'}| &= \left(\rho^{L'}|_{[n]}+\rho^{L'}|_{[2n+1]\setminus[n]}\right)+\left(\rho^{R'}|_{[n]}+\rho^{R'}|_{[2n+1]\setminus[n]}\right)\\
            &=\left(\rho^{L'}|_{[n]}+\rho^{R'}|_{[n]}\right)+\left(\rho^{L'}|_{[2n+1]\setminus[n]}+\rho^{R'}|_{[2n+1]\setminus[n]}\right)\\
            &\leq n+1 + |X_1|\\
            &= 2n+2.
        \end{align*}
        
        \begin{sloppypar}
        However, there exist permutations \(\pi\in \mathcal{S}_{2n+1}\), e.g., \(\pi^{L'}\) itself, such that \(\LCS(\pi, \pi^{L'}) +\LCS(\pi, \pi^{R'})=3n+2\), which contradicts the fact that \(\pi^{*'}\) is a median of \(\pi^{L'}\) and \(\pi^{R'}\). Therefore, there does not exist a symbol \(x\in[n]\) that appears in both \(\rho^{L'}\) and \(\rho^{R'}\).
        \end{sloppypar}

        Now, let \(A=\{a_1, a_2, \ldots , a_r\}\subseteq [n]\) with \(a_1<a_2<\cdots <a_r\) be the set of symbols that appear in \(\rho^{L'}|_{[n]}\). Similarly, let \(B=\{b_1, b_2, \ldots , b_{r'}\}\subseteq [n]\) with \(b_1>b_2>\cdots > b_{r'}\) be the set of symbols that appear in \(\rho^{R'}|_{[n]}\). We have just established that \(A\cap B = \emptyset\). We now claim that no symbol \(z\) of \(X_1\) appears to the left of any symbol in \(A\) in \(\pi^{*'}\). Indeed, if there existed such a symbol \(z\), then \(z\) would not appear in \(\rho^{L'}\) and we would have:
        \begin{align*}
            |\rho^{L'}|  + |\rho^{R'}| &= \left(\rho^{L'}|_{[n]}+\rho^{L'}|_{[2n+1]\setminus[n]}\right)+\left(\rho^{R'}|_{[n]}+\rho^{R'}|_{[2n+1]\setminus[n]}\right)\\
            &=\left(\rho^{L'}|_{[n]}+\rho^{R'}|_{[n]}\right)+\rho^{L'}|_{[2n+1]\setminus[n]}+\rho^{R'}|_{[2n+1]\setminus[n]}\\
            &\leq n + |X_1|-1 + |X_1|\\
            &= 3n+1.
        \end{align*}
        However, we have just seen examples of permutations with sum of LCS lengths \((3n+2)\), again contradicting the fact that \(\pi^{*'}\) is a median. Therefore, every symbol in \(A\) appears before every symbol of \(X_1\) in \(\pi^{*'}\). By a similar argument, every symbol in \(B\) appears after every symbol of \(X_1\) in \(\pi^{*'}\).

        Finally, we claim that \(A\cup B =[n]\). Assume otherwise. Then the set \(C:= [n]\setminus (A\cup B)\) is non-empty. We now modify \(\pi^{*'}\) in the following way -- first, for every symbol in \(C\) that appears to the right of any symbol of \(X_1\), we move it so that appears to the left of every symbol of \(X_1\), and then finally sort all symbols to the left of \(X_1\) in increasing order. It is not hard to see that after this procedure, \(|\rho^{L'}|\) increases by at least \(|C|\) while \(|\rho^{R'}|\) does not decrease since the symbols that were moved throughout the process never appeared in \(\rho^{R'}\) to begin with. Thus, we obtain a new permutation with strictly smaller median cost than \(\pi^{*'}\), once again reaching a contradiction. This completes the proof.
    \end{proof}

    By Lemmas~\ref{lemma:median-struct-1} and \ref{lemma:median-struct-2}, we can now conclude that \(\pi^*\) is necessarily in \(\mathcal{S}_N\).
\end{proof}

\section{From Multisets to Sets}
\label{sec:from-multisets-to-sets}

In this section, we show that \cm is at least as hard on sets as it is on multisets. In particular, we show the following lemma.

\begin{lemma}
    There is a polynomial time algorithm that takes a multiset \(S\subseteq \mathcal{S}_n\) with \(m\) permutations and produces a set \(S'\) of permutations satisfying the following properties.
    \begin{itemize}
        \item \(S' \subseteq \mathcal{S}_{n+2m}\), i.e., the length of each permutation in \(S'\) is \(n+2m\).
        \item \(|S'|=|S|\).
        \item For all \(k\), there exists \(\pi^{*}\in \mathcal{S}_n\) with \(\sum_{\pi\in S}d_U(\pi^*, \pi)\leq k\) if and only if there exists \(\pi^{*'}\in \mathcal{S}_{n+2m}\) with \(\sum_{\pi'\in S'}d_U(\pi^{*'}, \pi')\leq k+m\). 
    \end{itemize}
\end{lemma}
\begin{proof}
    Let \(\tau\in \mathcal{S}_{2m}\) be the identity permutation of length \(2m\). For each \(i\in [m]\), define \(\tau_i\) to be the permutation obtain from \(\tau\) by swapping the symbols \(2i\) and \(2i-1\). It is not hard to see that \(\tau\) is a  median of \(\{\tau_i\}_{i\in [m]}\) and for every permutation \(\tau'\in \mathcal{S}_{2m}\), \(\sum_{i\in [m]}d_U(\tau', \tau_i)\geq m\). (Indeed, this is clear if $\tau'$ differs from all permutations $\tau_i$, and otherwise, if $\tau' = \tau_1$, say, then $d_U(\tau', \tau_i) \geq 2$ for all $i \geq 2$.) To obtain \(S'\) from \(S\), we will append the permutation \(\tau_i\) to the \(i^{\text{th}}\) permutation of \(S\) for each \(i\in [m]\) using fresh symbols as necessary. More formally, if \(S=\{\pi_1, \pi_2, \ldots, \pi_m\}\), then:
    \[S':= \{\pi_i \circ \Delta_n(\tau_i) : i\in [m]\}.\]
    Now if there exists \(\pi^*\in \mathcal{S}_n\) with \(\sum_{\pi\in S}d_U(\pi^*, \pi)\leq k\), then we have:
    \begin{align*}
    \sum_{\pi\in S'}d_U(\pi^{*}\circ \Delta_n(\tau), \pi) &= \sum_{i\in [m]}d_U(\pi^{*}\circ \Delta_n(\tau), \pi_i\circ\Delta_n(\tau_i))\\
        &=\sum_{i\in[m]} d_U(\pi^*, \pi_i)+\sum_{i\in [m]}d_U(\tau, \tau_i)\\
        &\leq k+m.
    \end{align*}

For the other direction, assume there is some \(\pi^{*'}\) with \(\sum_{\pi'\in S'}d_U(\pi^{*'}, \pi')\leq k+m\). Define \(\pi^*:= \pi^{*'}|_{[n]}\) and \(\tau' = \pi^{*'}|_{[n+2m]\setminus [m]}\). For every \(\pi'\in S'\), we have \(d_U(\pi', \pi^{*'}) \geq d_U(\pi'|_{[n]}, \pi^*)+d_U(\pi'|_{[n+2m]\setminus [n]}, \tau')\), and consequently:
\[\sum_{i\in [m]}d_U(\pi_i, \pi^{*}) + \sum_{i\in [m]}d_U(\Delta_n(\tau_i), \tau')\leq k+m.\]
However, \(\sum_{i\in [m]}d_U(\Delta_n(\tau_i), \tau') \geq m\), and so we must have \(\sum_{i\in [m]}d_U(\pi_i, \pi^{*})\leq k\).
\end{proof}

\end{document}